\documentclass[sigconf]{acmart}

\def\BibTeX{{\rm B\kern-.05em{\sc i\kern-.025em b}\kern-.08emT\kern-.1667em\lower.7ex\hbox{E}\kern-.125emX}}

\renewcommand\footnotetextcopyrightpermission[1]{} 

\AtBeginDocument{%
  \providecommand\BibTeX{{%
    \normalfont B\kern-0.5em{\scshape i\kern-0.25em b}\kern-0.8em\TeX}}}

\usepackage{appendix}

\usepackage{bm}
\usepackage{xspace}
\usepackage{paralist}
\usepackage{graphicx}
\usepackage{epstopdf}
\usepackage{boxedminipage}
\usepackage{tabularx}
\usepackage{array}	
\usepackage{booktabs}
\usepackage{textcase}
\usepackage{multirow}

\usepackage{rotating}
\usepackage{algorithm}
\usepackage{algpseudocode}
\usepackage{caption}
\usepackage{enumitem}
\usepackage{amsmath,amsfonts,amssymb}
\usepackage{braket}

\usepackage{xcolor}
\definecolor{airforceblue}{rgb}{0.36, 0.54, 0.66}
\definecolor{bluegray}{rgb}{0.4, 0.6, 0.8}
\definecolor{ceil}{rgb}{0.57, 0.63, 0.81}
\definecolor{celestialblue}{rgb}{0.29, 0.59, 0.82}
\definecolor{cerulean}{rgb}{0.0, 0.48, 0.65}
\definecolor{celadon}{rgb}{0.67, 0.88, 0.69}
\definecolor{yellow-green}{rgb}{0.6, 0.8, 0.2}

\newtheorem{theorem}{Theorem}
\newtheorem{lemma}[theorem]{Lemma}
\newtheorem{corollary}[theorem]{Corollary}

\newtheorem{definition}[theorem]{Defn.}

\usepackage{graphicx,epstopdf}

\usepackage{subcaption}

\newcommand{\bA}{\boldsymbol{A}}

\newcommand{\Sec}[1]{\hyperref[sec:#1]{\S\ref*{sec:#1}}} 
\newcommand{\Eqn}[1]{\hyperref[eq:#1]{(\ref*{eq:#1})}} 
\newcommand{\Fig}[1]{\hyperref[fig:#1]{Fig.\,\ref*{fig:#1}}} 
\newcommand{\Tab}[1]{\hyperref[tab:#1]{Tab.\,\ref*{tab:#1}}} 
\newcommand{\Thm}[1]{\hyperref[thm:#1]{Theorem\,\ref*{thm:#1}}} 
\newcommand{\Fact}[1]{\hyperref[fact:#1]{Fact\,\ref*{fact:#1}}} 
\newcommand{\Lem}[1]{\hyperref[lem:#1]{Lemma\,\ref*{lem:#1}}} 
\newcommand{\Prop}[1]{\hyperref[prop:#1]{Prop.~\ref*{prop:#1}}} 
\newcommand{\Cor}[1]{\hyperref[cor:#1]{Corollary~\ref*{cor:#1}}} 
\newcommand{\Conj}[1]{\hyperref[conj:#1]{Conjecture~\ref*{conj:#1}}} 
\newcommand{\Def}[1]{\hyperref[def:#1]{Definition~\ref*{def:#1}}} 
\newcommand{\Alg}[1]{\hyperref[alg:#1]{Alg.~\ref*{alg:#1}}} 
\newcommand{\Ex}[1]{\hyperref[ex:#1]{Ex.~\ref*{ex:#1}}} 
\newcommand{\Clm}[1]{\hyperref[clm:#1]{Claim~\ref*{clm:#1}}} 
\newcommand{\Obs}[1]{\hyperref[obs:#1]{Obs.~\ref*{obs:#1}}} 

\newcommand{\cc}{D}

\newcommand{\dpath}{DP}
\newcommand{\dbp}{DBP}

\newcommand{\oMatch}{\mathcal{M}}

\newcommand{\oMatchSize}{\hbox{M}}
\newcommand{\uoMatchSize}{\hbox{DM}}

\newcommand{\uoIMatchSize}{\hbox{DIM}}

\newcommand{\frag}{\hbox{Frag}}
\newcommand{\shrink}{\hbox{Shrink}}
\newcommand{\numshrink}{\hbox{numSh}}
\newcommand{\sz}{\hbox{sz}}

\newcommand{\orb}{\hbox{orb}}

\usepackage{tikz}
\usetikzlibrary{decorations.markings,fit,arrows,positioning,backgrounds}
\tikzset{%
  gnode/.style={shape=circle,minimum size=3mm,fill,draw=black}
}
\tikzset{myptr/.style={decoration={markings,mark=at position 1 with {\arrow[scale=2,>=stealth]{>}}},postaction={decorate}}}

\def\Hzero{
	\node (1) at (0.5,0) [nd, fill=bluegray] {};
	\node (2) [label={left:\small 0}] at (0.5,0.5) [nd, fill=bluegray] {};
	\draw (1) to (2);
	\node [inner sep = 0,below] at +(0.5,-0.5) {{\Large $H_0$}};
}

\def\Hone{
	\node (1) at (0.5,0) [nd, fill=bluegray] {};
	\node (2) [label={left:\small 2}] at (0.5,0.5) [nd, fill=red!60] {};
	\node (3) [label={left:\small 1}] at (0.5,1) [nd, fill=bluegray] {};
	\draw (1) to (2);
	\draw (2) to (3);
	\node [inner sep = 0,below] at +(0.5,-0.5) {{\Large $H_1$}};
}

\def\Htwo{
	\node (1) at (0,0) [nd, fill=bluegray] {};
	\node (2) at (1,0) [nd, fill=bluegray] {};
	\node (3) [label={left:\small 3}] at (0.5,1) [nd, fill=bluegray] {};
	\draw (1) to (2);
	\draw (1) to (3);
	\draw (2) to (3);
	\node [inner sep = 0,below] at +(0.5,-0.5) {{\Large $H_2$}};
}

\def\Hthree{
	\node (1) at (0.5,0) [nd, fill=bluegray] {};
	\node (2) at (0.5,0.5) [nd, fill=red!60] {};
	\node (3) [label={left:\small 5}] at (0.5,1) [nd, fill=red!60] {};
	\node (4) [label={left:\small 4}] at (0.5,1.5) [nd, fill=bluegray] {};
	\draw (1) to (2);
	\draw (2) to (3);
	\draw (3) to (4);
	\node [inner sep = 0,below] at +(0.5,-0.5) {{\Large $H_3$}};
}

\def\Hfour{
	\node (1) at (0,0) [nd, fill=bluegray] {};
	\node (2) at (1,0) [nd, fill=bluegray] {};
	\node (3) [label={left:\small 7}] at (0.5,0.4) [nd, fill = red!60] {};
	\node (4) [label={left:\small 6}] at (0.5,1) [nd, fill=bluegray] {};
	\draw (1) to (3);
	\draw (2) to (3);
	\draw (3) to (4);
	\node [inner sep = 0,below] at +(0.5,-0.5) {{\Large $H_4$}};
}

\def\Hfive{
	\node (1) at (0,0) [nd, fill=bluegray] {};
	\node (2) at (1,0) [nd, fill=bluegray] {};
	\node (3) [label={left:\small 8}] at (0,1) [nd, fill=bluegray] {};
	\node (4) at (1,1) [nd, fill=bluegray] {};
	\draw (1) to (2);
	\draw (1) to (3);
	\draw (2) to (4);
	\draw (3) to (4);
	\node [inner sep = 0,below] at +(0.5,-0.5) {{\Large $H_5$}};
}

\def\Hsix{
	\node (1) [label={left:\small 9}] at (0.5,0) [nd, fill=yellow-green] {};
	\node (2) [label={left:\small 11}] at (0.5,0.5) [nd, fill=red!60] {};
	\node (3) [label={\small 10}] at (0,1) [nd, fill=bluegray] {};
	\node (4) at (1,1) [nd, fill=bluegray] {};
	\draw (1) to (2);
	\draw (2) to (3);
	\draw (2) to (4);
	\draw (3) to (4);
	\node [inner sep = 0,below] at +(0.5,-0.5) {{\Large $H_6$}};
}

\def\Hseven{
	\node (1) at (0.5,0) [nd, fill=red!60] {};
	\node (2) [label={left:\small 12}] at (0,0.5) [nd, fill=bluegray] {};
	\node (3) at (1,0.5) [nd, fill=bluegray] {};
	\node (4) [label={left:\small 13}] at (0.5,1) [nd, fill=red!60] {};
	\draw (1) to (2);
	\draw (1) to (3);
	\draw (1) to (4);
	\draw (2) to (4);
	\draw (3) to (4);
	\node [inner sep = 0,below] at +(0.5,-0.5) {{\Large $H_7$}};
}

\def\Height{
	\node (1) at (0,0) [nd, fill=bluegray] {};
	\node (2) at (1,0) [nd, fill=bluegray] {};
	\node (3) at (0.5,0.3) [nd, fill=bluegray] {};
	\node (4) [label={left:\small 14}] at (0.5,1) [nd, fill=bluegray] {};
	\draw (1) to (2);
	\draw (1) to (3);
	\draw (1) to (4);
	\draw (2) to (3);
	\draw (2) to (4);
	\draw (3) to (4);	
	\node [inner sep = 0,below] at +(0.5,-0.5) {{\Large $H_8$}};
}

\def\Hnine{
	\node (1) at (0.5,0) [nd, fill=bluegray] {};
	\node (2) at (0.5,0.5) [nd, fill=red!60] {};
	\node (3) [label={left:\small 17}] at (0.5,1) [nd, fill=yellow-green] {};
	\node (4) [label={left:\small 16}] at (0.5,1.5) [nd, fill=red!60] {};
	\node (5) [label={left:\small 15}] at (0.5,2) [nd, fill=bluegray] {};
	\draw (1) to (2);
	\draw (2) to (3);
	\draw (3) to (4);
	\draw (4) to (5);
	\node [inner sep = 0,below] at +(0.5,-0.5) {{\Large $H_9$}};
}

\def\Hten{
	\node (1) [label={right:\small 19}] at (0,0) [nd, fill=bluegray] {};
	\node (2) at (1,0) [nd, fill=bluegray] {};
	\node (3) [label={left:\small 21}] at (0.5,0.5) [nd, fill=red!60] {};
	\node (4) [label={left:\small 20}] at (0.5,1) [nd, fill=yellow-green] {};
	\node (5) [label={left:\small 18}] at (0.5, 1.5) [nd, fill=yellow] {};
	\draw (1) to (3);
	\draw (2) to (3);
	\draw (3) to (4);
	\draw (4) to (5);
	\node [inner sep = 0,below] at +(0.5,-0.5) {{\Large $H_{10}$}};
}

\def\Heleven{
	\node (1) at (0.5,0) [nd, fill=bluegray] {};
	\node (2) at (0,0.5) [nd, fill=bluegray] {};
	\node (3) [label={[label distance=0.01cm]135: \small 23}] at (0.5,0.5) [nd, fill=red!60] {};
	\node (4) at (1,0.5) [nd, fill=bluegray] {};
	\node (5) [label={\small 22}] at (0.5, 1) [nd, fill=bluegray] {};
	\draw (1) to (3);
	\draw (2) to (3);
	\draw (3) to (4);
	\draw (3) to (5);
	\node [inner sep = 0,below] at +(0.5,-0.5) {{\Large $H_{11}$}};
}

\def\Htwelve{
	\node (1) [label={right:\small 24}] at (0,0) [nd, fill=yellow-green] {};
	\node (2) at (1,0) [nd, fill=yellow-green] {};
	\node (3) [label={\small 26}] at (0,0.6) [nd, fill=bluegray] {};
	\node (4) at (1,0.6) [nd, fill=bluegray] {};
	\node (5) [label={\small 25}] at (0.5,1.1) [nd, fill=red!60] {};
	\draw (1) to (3);
	\draw (2) to (4);
	\draw (3) to (4);
	\draw (3) to (5);
	\draw (4) to (5);
	\node [inner sep = 0,below] at +(0.5,-0.5) {{\Large $H_{12}$}};
}

\def\Hthirteen{
	\node (1) [label={[label distance=0.01cm]135:\small 29}] at (0,0) [nd, fill=bluegray] {};
	\node (2) at (1,0) [nd, fill=bluegray] {};
	\node (3) [label={left:\small 30}] at (0.5,0.5) [nd, fill=red!60] {};
	\node (4) [label={left:\small 28}] at (0.5,1) [nd, fill=yellow-green] {};
	\node (5) [label={left:\small 27}] at (0.5, 1.5) [nd, fill=yellow] {};
	\draw (1) to (2);
	\draw (1) to (3);
	\draw (2) to (3);
	\draw (3) to (4);
	\draw (4) to (5);
	\node [inner sep = 0,below] at +(0.5,-0.5) {{\Large $H_{13}$}};
}

\def\Hfourteen{
	\node (1) [label={right:\small 31}] at (0,0) [nd, fill=bluegray] {};
	\node (2) at (1,0) [nd, fill=bluegray] {};
	\node (3) [label={left:\small 33}] at (0.5,0.5) [nd, fill=red!60] {};
	\node (4) [label={\small 32}] at (0,1) [nd, fill=yellow-green] {};
	\node (5) at (1, 1) [nd, fill=yellow-green] {};
	\draw (1) to (3);
	\draw (2) to (3);
	\draw (3) to (4);
	\draw (3) to (5);
	\draw (4) to (5);
	\node [inner sep = 0,below] at +(0.5,-0.5) {{\Large $H_{14}$}};
}

\def\Hfifteen{
	\node (1) at (0.2,0) [nd, fill=bluegray] {};
	\node (2) at (0.8,0) [nd, fill=bluegray] {};
	\node (3) at (0,0.6) [nd, fill=bluegray] {};
	\node (4) at (1,0.6) [nd, fill=bluegray] {};
	\node (5) [label={\small 34}] at (0.5, 1) [nd, fill=bluegray] {};
	\draw (1) to (2);
	\draw (1) to (3);
	\draw (2) to (4);
	\draw (3) to (5);
	\draw (4) to (5);
	\node [inner sep = 0,below] at +(0.5,-0.5) {{\Large $H_{15}$}};
}

\def\Hsixteen{
	\node (1) [label={left:\small 35}] at (0.5,0) [nd, fill=yellow-green] {};
	\node (2) [label={left:\small 38}] at (0.5,0.6) [nd, fill=red!60] {};
	\node (3) [label={right:\small 37}] at (0,1.1) [nd, fill=bluegray] {};
	\node (4) at (1,1.1) [nd, fill=bluegray] {};
	\node (5) [label={left:\small 36}] at (0.5, 1.6) [nd, fill=yellow] {};
	\draw (1) to (2);
	\draw (2) to (3);
	\draw (2) to (4);
	\draw (3) to (5);
	\draw (4) to (5);
	\node [inner sep = 0,below] at +(0.5,-0.5) {{\Large $H_{16}$}};
}

\def\Hseventeen{
	\node (1) [label={left:\small 39}] at (0.5,0) [nd, fill=yellow-green] {};
	\node (2) [label={left:\small 42}] at (0.5,0.6) [nd, fill=red!60] {};
	\node (3) [label={left:\small 40}] at (0,1.1) [nd, fill=bluegray] {};
	\node (4) at (1,1.1) [nd, fill=bluegray] {};
	\node (5) [label={left:\small 41}] at (0.5, 1.6) [nd, fill=yellow] {};
	\draw (1) to (2);
	\draw (2) to (3);
	\draw (2) to (4);
	\draw (2) to (5);
	\draw (3) to (5);
	\draw (4) to (5);
	\node [inner sep = 0,below] at +(0.5,-0.5) {{\Large $H_{17}$}};
}

\def\Heighteen{
	\node (1) at (0,0) [nd, fill=bluegray] {};
	\node (2) at (1,0) [nd, fill=bluegray] {};
	\node (3) [label={left:\small 44}] at (0.5,0.5) [nd, fill=red!60] {};
	\node (4) [label={\small 43}] at (0,1) [nd, fill=bluegray] {};
	\node (5) at (1, 1) [nd, fill=bluegray] {};
	\draw (1) to (2);
	\draw (1) to (3);
	\draw (2) to (3);
	\draw (3) to (4);
	\draw (3) to (5);
	\draw (4) to (5);
	\node [inner sep = 0,below] at +(0.5,-0.5) {{\Large $H_{18}$}};
}

\def\Hnineteen{
	\node (1) [label={left:\small 45}] at (0.5,0) [nd, fill=yellow-green] {};
	\node (2) [label={left:\small 47}] at (0.5,0.6) [nd, fill=red!60] {};
	\node (3) [label={\small 48}] at (0,1.1) [nd, fill=bluegray] {};
	\node (4) at (1,1.1) [nd, fill=bluegray] {};
	\node (5) [label={\small 46}] at (0.5, 1.6) [nd, fill=yellow] {};
	\draw (1) to (2);
	\draw (2) to (3);
	\draw (2) to (4);
	\draw (3) to (4);
	\draw (3) to (5);
	\draw (4) to (5);
	\node [inner sep = 0,below] at +(0.5,-0.5) {{\Large $H_{19}$}};
}

\def\Htwenty{
	\node [label={left:\small 50}] (1) at (0.5,0) [nd, fill=red!60] {};
	\node [label={\small 49}] (2) at (0,0.5) [nd, fill=bluegray] {};
	\node (3) at (0.5,0.5) [nd, fill=bluegray] {};
	\node (4) at (1,0.5) [nd, fill=bluegray] {};
	\node (5) at (0.5, 1) [nd, fill=red!60] {};
	\draw (1) to (2);
	\draw (1) to (3);
	\draw (1) to (4);
	\draw (2) to (5);
	\draw (3) to (5);
	\draw (4) to (5);
	\node [inner sep = 0,below] at +(0.5,-0.5) {{\Large $H_{20}$}};
}

\def\Htwentyone{
	\node (1) [label={left:\small 51}] at (0,0) [nd, fill=yellow-green] {};
	\node (2) at (1,0) [nd, fill=yellow-green] {};
	\node (3) [label={left:\small 53}] at (0,0.6) [nd, fill=bluegray] {};
	\node (4) at (1,0.6) [nd, fill=bluegray] {};
	\node (5) [label={left:\small 52}] at (0.5,1.1) [nd, fill=red!60] {};
	\draw (1) to (2);
	\draw (1) to (3);
	\draw (2) to (4);
	\draw (3) to (4);
	\draw (3) to (5);
	\draw (4) to (5);
	\node [inner sep = 0,below] at +(0.5,-0.5) {{\Large $H_{21}$}};
}

\def\Htwentytwo{
	\node (1) [label={left:\small 54}] at (0.5,0) [nd, fill=yellow-green] {};
	\node (2) at (0.5,0.6) [nd, fill=yellow-green] {};
	\node (3) [label={\small 55}] at (0,1.1) [nd, fill=bluegray] {};
	\node (4) at (1,1.1) [nd, fill=bluegray] {};
	\node (5) at (0.5,1.6) [nd, fill=red!60] {};
	\draw (1) to (3);
	\draw (1) to (4);
	\draw (2) to (3);
	\draw (2) to (4);
	\draw (3) to (4);
	\draw (3) to (5);
	S\draw (4) to (5);
	\node [inner sep = 0,below] at +(0.5,-0.5) {{\Large $H_{22}$}};
}

\def\Htwentythree{
	\node (1) [label={left:\small 56}]  at (0.5,0) [nd, fill=yellow-green] {};
	\node (2) [label={left:\small 58}] at (0.5,0.6) [nd, fill=red!60] {};
	\node (3) [label={\small 57}] at (0,1.1) [nd, fill=bluegray] {};
	\node (4) at (1,1.1) [nd, fill=bluegray] {};
	\node (5) at (0.5, 1.6) [nd, fill=bluegray] {};
	\draw (1) to (2);
	\draw (2) to (3);
	\draw (2) to (4);
	\draw (2) to (5);
	\draw (3) to (4);
	\draw (3) to (5);
	\draw (4) to (5);
	\node [inner sep = 0,below] at +(0.5,-0.5) {{\Large $H_{23}$}};
}

\def\Htwentyfour{
	\node (1) [label={left:\small 59}] at (0,0) [nd, fill=bluegray] {};
	\node (2) [label={below:\small 60}] at (0.6,0.3) [nd, fill=yellow-green] {};
	\node (3) [label={left:\small 61}] at (0,0.6) [nd, fill=red!60] {};
	\node (4) at (0.6,0.9) [nd, fill=yellow-green] {};
	\node (5) at (0,1.2) [nd, fill=bluegray] {};
	\draw (1) to (2) [nd] {};
	\draw (1) to (3) [nd] {};
	\draw (2) to (3) [nd] {};
	\draw (2) to (4) [nd] {};
	\draw (3) to (4) [nd] {};
	\draw (3) to (5) [nd] {};
	\draw (4) to (5) [nd] {};
	\node [inner sep = 0,below] at +(0.5,-0.5) {{\Large $H_{24}$}};
}

\def\Htwentyfive{
	\node (1) at (0,0) [nd, fill=red!60] {};
	\node (2) [label={right:\small 62}]at (1,0) [nd, fill=yellow-green] {};
	\node (3) [label={left:\small 64}] at (0.5,0.5) [nd, fill=bluegray] {};
	\node (4) at (0,1) [nd, fill=bluegray] {};
	\node (5) [label={\small 63}] at (1,1) [nd, fill=red!60] {};
	
	\draw (1) to (2);
	\draw (1) to (3);
	\draw (1) to (4);
	\draw (2) to (5);
	\draw (3) to (4);
	\draw (3) to (5);
	\draw (4) to (5);
	\node [inner sep = 0,below] at +(0.5,-0.5) {{\Large $H_{25}$}};
}

\def\Htwentysix{
	\node (1) [label={left:\small 65}] at (0.5,0) [nd, fill=yellow-green] {};
	\node (2) at (0,0.5) [nd, fill=red!60] {};
	\node (3) [label={below:\small 67}] at (1,0.5) [nd, fill=red!60] {};
	\node (4) at (0.5,1) [nd, fill=bluegray] {};
	\node (5) [label={left:\small 66}] at (0.5,1.6) [nd, fill=bluegray] {};
	\draw (1) to (2);
	\draw (1) to (3);
	\draw (2) to (3);
	\draw (2) to (4);
	\draw (2) to (5);
	\draw (3) to (4);
	\draw (3) to (5);
	\draw (4) to (5);
	\node [inner sep = 0,below] at +(0.5,-0.5) {{\Large $H_{26}$}};
}

\def\Htwentyseven{
	\node (1) at (0,0) [nd, fill=bluegray] {};
	\node (2) at (1,0) [nd, fill=bluegray] {};
	\node (3) [label={left:\small 69}] at (0.5,0.5) [nd, fill=red!60] {};
	\node (4) [label={\small 68}] at (0,1) [nd, fill=bluegray] {};
	\node (5) at (1,1) [nd, fill=bluegray] {};
	
	\draw (1) to (2);
	\draw (1) to (3);
	\draw (1) to (4);
	\draw (2) to (3);
	\draw (2) to (5);
	\draw (3) to (4);
	\draw (3) to (5);
	\draw (4) to (5);
	\node [inner sep = 0,below] at +(0.5,-0.5) {{\Large $H_{27}$}};
}

\def\Htwentyeight{
	\node (1) [label={left:\small 70}] at (0.5,0) [nd, fill=red!60] {};
	\node (2) at (0.5,0.6) [nd, fill=bluegray] {};
	\node (3) at (0,1.1) [nd, fill=bluegray] {};
	\node (4) [label={\small 71}] at (1,1.1) [nd, fill=bluegray] {};
	\node (5) at (0.5,1.6) [nd, fill=red!60] {};
	\draw (1) to (2);
	\draw (1) to (3);
	\draw (1) to (4);
	\draw (2) to (3);
	\draw (2) to (4);
	\draw (2) to (5);
	\draw (3) to (4);
	\draw (3) to (5);
	\draw (4) to (5);
	\node [inner sep = 0,below] at +(0.5,-0.5) {{\Large $H_{28}$}};
}

\def\Htwentynine{
	\node (1) at (0.2,0) [nd, fill=bluegray] {};
	\node (2) at (0.8,0) [nd, fill=bluegray] {};
	\node (3) at (0,0.6) [nd, fill=bluegray] {};
	\node (4) at (1,0.6) [nd, fill=bluegray] {};
	\node (5) [label={\small 72}] at (0.5, 1) [nd, fill=bluegray] {};
	\draw (1) to (2);
	\draw (1) to (3);
	\draw (1) to (4);
	\draw (1) to (5);
	\draw (2) to (3);
	\draw (2) to (4);
	\draw (2) to (5);
	\draw (3) to (4);
	\draw (3) to (5);
	\draw (4) to (5);
	\node [inner sep = 0,below] at +(0.5,-0.5) {{\Large $H_{29}$}};
}

\def\HoneEdge{
	\node (1) at (0.5,0) [nd, fill=bluegray] {};
	\node (2) at (0.5,0.5) [nd, fill=red!60] {};
	\node (3) at (0.5,1) [nd, fill=bluegray] {};
	\draw (1) to (2);
	\draw (2) to (3);
	\node [inner sep = 0,left] at +(0.4,0.75) {0};
	\node [inner sep = 0,below] at +(0.5,-0.5) {{\Large $H_1$}};
}

\def\HtwoEdge{
	\node (1) at (0,0) [nd, fill=bluegray] {};
	\node (2) at (1,0) [nd, fill=bluegray] {};
	\node (3) at (0.5,1) [nd, fill=bluegray] {};
	\draw (1) to (2);
	\draw (1) to (3);
	\draw (2) to (3);
	\node [rotate = 68,inner sep = 0,left] at +(0.15,0.75) {1};
	\node [inner sep = 0,below] at +(0.5,-0.5) {{\Large $H_2$}};
}

\def\HthreeEdge{
	\node (1) at (0.5,0) [nd, fill=bluegray] {};
	\node (2) at (0.5,0.5) [nd, fill=red!60] {};
	\node (3) at (0.5,1) [nd, fill=red!60] {};
	\node (4) at (0.5,1.5) [nd, fill=bluegray] {};
	\draw (1) to (2);
	\draw[densely dotted] (2) to (3);
	\draw (3) to (4);
	\node [rotate = 90,inner sep = 0,left] at +(0.25,0.85) {3};
	\node [rotate = 90,inner sep = 0,left] at +(0.25,1.35) {2};
	\node [inner sep = 0,below] at +(0.5,-0.5) {{\Large $H_3$}};
}

\def\HfourEdge{
	\node (1) at (0,0) [nd, fill=bluegray] {};
	\node (2) at (1,0) [nd, fill=bluegray] {};
	\node (3) at (0.5,0.4) [nd, fill = red!60] {};
	\node (4) at (0.5,1) [nd, fill=bluegray] {};
	\draw (1) to (3);
	\draw (2) to (3);
	\draw (3) to (4);
	\node [rotate = 90,inner sep = 0,left] at +(0.3,0.8) {4};
	\node [inner sep = 0,below] at +(0.5,-0.5) {{\Large $H_4$}};
}

\def\HfiveEdge{
	\node (1) at (0,0) [nd, fill=bluegray] {};
	\node (2) at (1,0) [nd, fill=bluegray] {};
	\node (3) at (0,1) [nd, fill=bluegray] {};
	\node (4) at (1,1) [nd, fill=bluegray] {};
	\draw (1) to (2);
	\draw (1) to (3);
	\draw (2) to (4);
	\draw (3) to (4);
	\node [inner sep = 0,above] at +(0.5,1.1) {5};
	\node [inner sep = 0,below] at +(0.5,-0.5) {{\Large $H_5$}};
}

\def\HsixEdge{
	\node (1) at (0.5,0) [nd, fill=yellow-green] {};
	\node (2) at (0.5,0.5) [nd, fill=red!60] {};
	\node (3) at (0,1) [nd, fill=bluegray] {};
	\node (4) at (1,1) [nd, fill=bluegray] {};
	\draw (1) to (2);
	\draw [densely dotted] (2) to (3);
	\draw [densely dotted] (2) to (4);
	\draw [line width=1.4 pt] (3) to (4);
	\node [rotate = 90,inner sep = 0,left] at +(0.3,0.4) {6};
	\node [inner sep = 0,above] at +(0.5,1.1) {7};
	\node [rotate = -45,inner sep = 0,right] at +(0.35,0.9) {8};
	\node [inner sep = 0,below] at +(0.5,-0.5) {{\Large $H_6$}};
}

\def\HsevenEdge{
	\node (1) at (0.5,0) [nd, fill=red!60] {};
	\node (2) at (0,0.5) [nd, fill=bluegray] {};
	\node (3) at (1,0.5) [nd, fill=bluegray] {};
	\node (4) at (0.5,1) [nd, fill=red!60] {};
	\draw (1) to (2);
	\draw (1) to (3);
	\draw [densely dotted] (1) to (4);
	\draw (2) to (4);
	\draw (3) to (4);
	\node [rotate = 68,inner sep = 0,right] at +(0.1,0.82) {9};
	\node [rotate = 90,inner sep = 0,right] at +(0.36,0.35) {10};
	\node [inner sep = 0,below] at +(0.5,-0.5) {{\Large $H_7$}};
}

\def\HeightEdge{
	\node (1) at (0,0) [nd, fill=bluegray] {};
	\node (2) at (1,0) [nd, fill=bluegray] {};
	\node (3) at (0.5,0.3) [nd, fill=bluegray] {};
	\node (4) at (0.5,1) [nd, fill=bluegray] {};
	\draw (1) to (2);
	\draw (1) to (3);
	\draw (1) to (4);
	\draw (2) to (3);
	\draw (2) to (4);
	\draw (3) to (4);
	\node [rotate = 68, inner sep = 0,right] at +(0.08,0.5) {11};	
	\node [inner sep = 0,below] at +(0.5,-0.5) {{\Large $H_8$}};
}

\def\DirThreePath{
   \node (1) at (0.5,-0.2) [nd] {};
   \node (2) at (-0.1,0.55) [nd] {};
   \node (3) at (1.1,0.55) [nd] {};
   \node (4) at (0.1,1.3) [nd] {};
   \draw[myptr] (1) to (2);
   \draw[myptr] (1) to (3);
   \draw[myptr] (2) to (4);
\node [inner sep = 0,below] at +(0.5,-0.5) {{\Large Directed 3-path}};
}

\def\Diamond{
   \node (1) at (0,0) [nd] {};
   \node (2) at (0,1) [nd] {};
   \node (3) at (1,0) [nd] {};
   \node (4) at (1,1) [nd] {};
   \draw (1) to (2);
   \draw (1) to (3);
   \draw (1) to (4);
   \draw (2) to (4);
   \draw (3) to (4);      
  \node [inner sep = 0,below] at +(0.5,-0.5) {{\Large Diamond}};
}

\def\DirBiPyramid{
   \node (1) at (0.5,-0.2) [nd] {};
   \node (2) at (-0.1,0.55) [nd] {};
   \node (3) at (1.1,0.55) [nd] {};
   \node (4) at (0.1,1.3) [nd] {};
   \node (5) at (0.9,1.3) [nd] {};
   \draw[-stealth'] (1) to (2);
   \draw[-stealth'] (1) to (3);
   \draw[-stealth'] (1) to (4);
   \draw[-stealth'] (1) to (5);
   \draw[-stealth'] (2) to (4);
   \draw[-stealth'] (2) to (5);
   \draw[-stealth'] (3) to (2);
   \draw[-stealth'] (3) to (4);
   \draw[-stealth'] (5) to (3);   
\node [inner sep = 0,below] at +(0.5,-0.5) {{\Large Directed bipyramid}};
}

\def\Wedge {
   \node (1) at (0.75,0) [nd] {};
   \node (2) at (0,1) [nd] {};
   \node (3) at (1.5,1) [nd] {};
   \draw (1) to (2);
   \draw (1) to (3);    
  \node [inner sep = 0,below] at +(0.75,-0.5) {{\Large Wedge}};
}

\def\DirFiveCycleOne{
   \node (1) at (0.5,2) [nd] {$w$};   
   \node (2) at (1.1,1) [nd, fill=white] {$i$};   
   \node (3) at (0.9,0) [nd] {$j$};
   \node (4) at (0.1,0) [nd] {$k$};
   \node (5) at (-0.1,1) [nd] {$l$};

   \draw[-stealth'] (1) to (2);
   \draw[-stealth'] (3) to (2);
   \draw[-stealth'] (4) to (3);
   \draw[-stealth'] (4) to (5);
   \draw[-stealth'] (1) to (5);
\node [inner sep = 0,below] at +(0.5,-0.5) {{\Large (a)}};
}

\def\DirFiveCycleTwo{
   \node (1) at (0.5,2) [nd] {$k$};   
   \node (2) at (1.1,1) [nd] {$j$};   
   \node (3) at (0.9,0) [nd, fill=white] {$i$};
   \node (4) at (0.1,0) [nd] {$w$};
   \node (5) at (-0.1,1) [nd] {$l$};

   \draw[-stealth'] (1) to (2);
   \draw[-stealth'] (2) to (3);
   \draw[-stealth'] (4) to (3);
   \draw[-stealth'] (5) to (4);
   \draw[-stealth'] (1) to (5);
\node [inner sep = 0,below] at +(0.5,-0.5) {{\Large (b)}};
}

\def\DirFiveCycleThree{
   \node (1) at (0.5,2) [nd] {$k$};   
   \node (2) at (1.1,1) [nd] {$l$};   
   \node (3) at (0.9,0) [nd] {$w$};
   \node (4) at (0.1,0) [nd, fill=white] {$i$};
   \node (5) at (-0.1,1) [nd] {$j$};

   \draw[-stealth'] (1) to (2);
   \draw[-stealth'] (3) to (2);
   \draw[-stealth'] (4) to (3);
   \draw[-stealth'] (5) to (4);
   \draw[-stealth'] (1) to (5);
\node [inner sep = 0,below] at +(0.5,-0.5) {{\Large (c)}};
}

\def\DirThreePathForFiveCycle{
   \node (1) at (0.1,2) [nd] {$i$};
   \node (3) at (-0.1,1) [nd] {$j$};
   \node[fill=gray] (4) at (1.1,1) [nd] {$l$};
   \node (5) at (0.5,0) [nd] {$k$};
   \draw[-stealth'] (5) to (3);
   \draw[-stealth'] (5) to (4);
   \draw[-stealth'] (3) to (1);
\node [inner sep = 0,below] at +(0.5,-0.5) {{\Large Directed 3-path}};
}

\def\DirTailedTriOne{
   \node (1) at (0.1,2) [nd] {$i$};
   \node (3) at (-0.1,1) [nd] {$j$};
   \node[fill=gray] (4) at (1.1,1) [nd] {$l$};
   \node (5) at (0.5,0) [nd] {$k$};
   \draw[myptr] (5) to (3);
   \draw[myptr] (5) to (4);
   \draw[myptr] (4) to (3);
   \draw[myptr] (3) to (1);
\node [inner sep = 0,below] at +(0.5,-0.5) {{\Large (1)}};
}

\def\DirTailedTriTwo{
   \node (1) at (0.1,2) [nd] {$i$};
   \node (3) at (-0.1,1) [nd] {$j$};
   \node[fill=gray] (4) at (1.1,1) [nd] {$l$};
   \node (5) at (0.5,0) [nd] {$k$};
   \draw[myptr] (5) to (3);
   \draw[myptr] (5) to (4);
   \draw[myptr] (3) to (4);
   \draw[myptr] (3) to (1);
\node [inner sep = 0,below] at +(0.5,-0.5) {{\Large (2)}};
}

\def\DirTailedTriThree{
   \node (1) at (0.1,2) [nd] {$l$};
   \node (3) at (-0.1,1) [nd] {$k$};
   \node[fill=gray] (4) at (1.1,1) [nd] {$i$};
   \node (5) at (0.5,0) [nd] {$j$};
   \draw[myptr] (3) to (5);
   \draw[myptr] (5) to (4);
   \draw[myptr] (3) to (4);
   \draw[myptr] (3) to (1);
\node [inner sep = 0,below] at +(0.5,-0.5) {{\Large (3)}};
}

\begin{document}

\pagestyle{plain}

\title{Efficiently Counting Vertex Orbits of All 5-vertex Subgraphs, by EVOKE}

\author{Noujan Pashanasangi}
\affiliation{%
  \institution{University of California, Santa Cruz}
  \city{Santa Cruz}
  \state{CA}}
\email{npashana@ucsc.edu}
\authornote{Both authors are supported by NSF Awards CCF-1740850, CCF-1813165, and ARO Award W911NF1910294.}

\author{C. Seshadhri}
\affiliation{%
  \institution{University of California, Santa Cruz}
  \city{Santa Cruz}
  \state{CA}}
\email{sesh@ucsc.edu}
\authornotemark[1]

\begin{abstract}
Subgraph counting is a fundamental task in network analysis. Typically, algorithmic
work is on total counting, where we wish to count the total frequency 
of a (small) pattern subgraph in a large input data set. But many applications require
\emph{local counts} (also called vertex orbit counts) 
wherein, for every vertex $v$ of the input graph, one needs
the count of the pattern subgraph involving $v$. This provides a rich set
of vertex features that can be used in machine learning tasks, especially classification
and clustering. But getting local counts is extremely challenging. Even the easier problem
of getting total counts has received much research attention. Local counts
require algorithms that get much finer grained information, and the sheer output
size makes it difficult to design scalable algorithms.

We present EVOKE, a scalable algorithm that can determine
vertex orbits counts for all 5-vertex pattern subgraphs. In other words,
EVOKE exactly determines, for every vertex $v$ of the input graph
and every 5-vertex subgraph $H$, the number of copies of $H$ that $v$ participates in.
EVOKE can process graphs with tens of millions of edges, within an hour on a commodity
machine. EVOKE is typically hundreds of times faster than previous state of the art algorithms,
and gets results on datasets beyond the reach of previous methods.

Theoretically, we generalize a recent ``graph cutting'' framework to get vertex orbit counts.
This framework generate a collection of polynomial equations relating vertex orbit counts
of larger subgraphs to those of smaller subgraphs. EVOKE carefully exploits the structure
among these equations to rapidly count. We prove and empirically validate that EVOKE
only has a small constant factor overhead over the best (total) 5-vertex subgraph counter.
\end{abstract}


\maketitle

\section{Introduction}
\label{sec:introduction}
One of the most important algorithmic techniques
in network analysis is \emph{subgraph counting}, also referred to as motif counting or graphlet analysis.
Subgraph counting is basically the problem of counting the frequency of small pattern subgraphs
in a large input graph.  
These techniques have found applications in bioinformatics and biological networks~\cite{Pr07, HoBe+07, PrzuljCJ04}, 
social networks~\cite{o2012identifying,UganderBK13, Fa10, SzTh10, SonKanKim12}, community and dense subgraph detection~\cite{SaSe15, BeGlLe16, Ts15,TPM17}, social sciences~\cite{Burt04, Co88, HoLe70, Po98}, 
and many other applications~\cite{Fa07, BeBoCaGi08, GonenS09, FoDeCo10, agrawal2018large}. (Refer to the tutorial~\cite{SeTi19} for more details on applications.)

Let $G$ denote the input graph, that we wish to analyze.
While the typical description of subgraph counting asks for the total count
of a pattern subgraph in $G$, many applications require \emph{local} counts.
(These are also referred to as graphlet distributions, orbit counts, or $k$-profiles.)
For a given set of patterns, the aim is to find, \emph{for every vertex $v$ of $G$},
the number of patterns that $v$ participates in. This is a much finer grained description
of the graph, and can be used to generate features for vertices. A compelling
application of these local counts are the \emph{graphlet kernel}, where local counts
are used to construct vector representations of vertices for machine learning~\cite{ShervashidzeVPMB09}. In many applications (documented in \Sec{related}), 
one typically wants local counts for all pattern subgraphs of up to a given size.

Subgraph counting is an extremely challenging problem. As shown in previous work,
even for a moderate graph with a few million edges, counts of (say) $5$-vertex
pattern subgraphs can be in the order of billions to trillions~\cite{JhSePi15,AhNe+15,PiSeVi17}. 
This combinatorial explosion is often tamed by clever counting methods that avoid
enumeration, but these are tailored to global counts in $G$. There has been
recent work on randomized methods for local counting, but these require large
parallel hardware even for graphs with tens of millions of edges~\cite{ElShBo16}.

\subsection{Problem Description} \label{sec:desc}

The input $G = (V,E)$ is a simple, undirected graph. Our aim is to get local counts,
for every vertex in $G$, for all the patterns given in \Fig{node_orbits}. 
\Fig{node_orbits} shows all connected subgraphs with at most $5$ vertices. We will refer to these
as \emph{patterns}. (We do not focus on disconnected
patterns; results in~\cite{PiSeVi17} imply that these can be easily determined 
from connected subgraph counts.) Within each pattern,
vertices are present in different ``roles'' or \emph{orbits}. In some patterns
like the 5-cycle ($H_{15}$) and 5-clique ($H_{29}$), there is just one orbit.
In contrast, $H_{10}$ has four different orbits, indicated by the different
colors. Thus, a vertex of $G$
can participate in a copy of $H_{10}$ in four different ways, and we wish to determine
all of these four different counts.
We delay the exact formalism of orbits to \Sec{prelim}.
But hopefully, \Fig{node_orbits} gives
a clear pictorial representation of the 73 different orbits,
numbered individually. 

Our aim is to design an algorithm that: for every vertex $v$ in $V$ and every orbit $\theta$, 
exactly outputs the number of times that $v$ occurs in a copy of $\theta$. Thus, the output
is a set of $73|V|$ counts. (Technically, we ask for induced counts, but can
also get non-induced counts. Details in \Sec{prelim}.)
For example, the count of orbit $17$ is the number
of times that $v$ is the middle of a 4-path, while the count of orbit $15$
is the number of 4-paths that start/end at $v$. Analogously, the count of orbit $34$
is the number of 5-cycles that $v$ participates in.
For a fixed orbit, we refer to these numbers as the \emph{vertex orbit counts} (VOCs).
Collectively (over all orbits), we wish to determine \emph{VOCs
for all $5$-vertex subgraphs}. For convenience, we refer to this as simply $5$-VOCs.
We refer to the total subgraph count as ``global'' counts, which is clearly a
much easier problem.

As can be seen, the desired output is an immensely rich local description of the vertices
of $G$. This output subsumes a number of recent
subgraph counting problems in the data mining community~\cite{AhNe+15,ElShBo15,ElShBo16,PiSeVi17}.

{\bf Main challenges:} To the best of our knowledge,
there is no algorithm that (even approximately) computes all 5-VOCs
even for graphs with tens of millions of edges.
Results on global counting are much faster, but it is not clear how to implement
these ideas for VOCs~\cite{AhNe+15,PiSeVi17}. 
The ORCA package is the only algorithm that actually computes all 5-VOCs,
but it does not terminate after days for graphs with tens of millions of edges.
We give more details of previous work in \Sec{related}. 

From a mathematical standpoint, the challenge is to get all 5-VOCs
without an expensive enumeration. The total number of orbit counts is easily in the order
of trillions, and a fast algorithm should ideally avoid touching each 5-vertex subgraph in $G$. 
On the other hand, VOCs are an extremely fine-grained statistic, so purely 
global methods do not work.

\begin{figure}[th]
\centering
 \includegraphics[width=0.9\columnwidth]{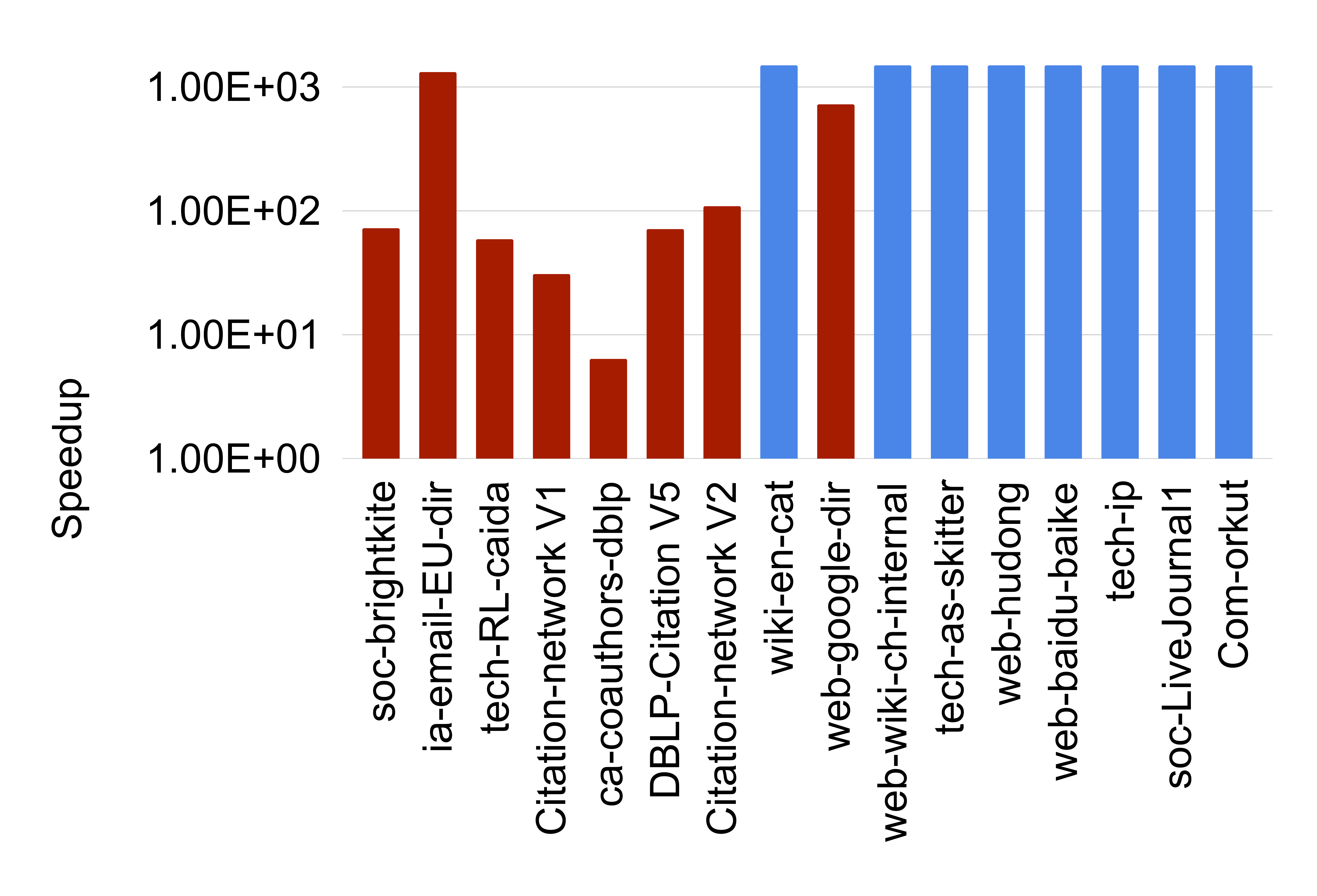}
 \caption{\label{fig:speedup} \small Runtime speedup for computing all 5-VOCs achieved by EVOKE over ORCA (computed as runtime of ORCA/runtime of EVOKE). Graphs are sorted by increasing number of edges from left to right. For the blue bars, ORCA ran out of memory or did not terminate after
 1000 times the EVOKE running time. EVOKE is significantly faster than ORCA, and makes
 5-VOC counting feasible for large graphs.} 
\end{figure}

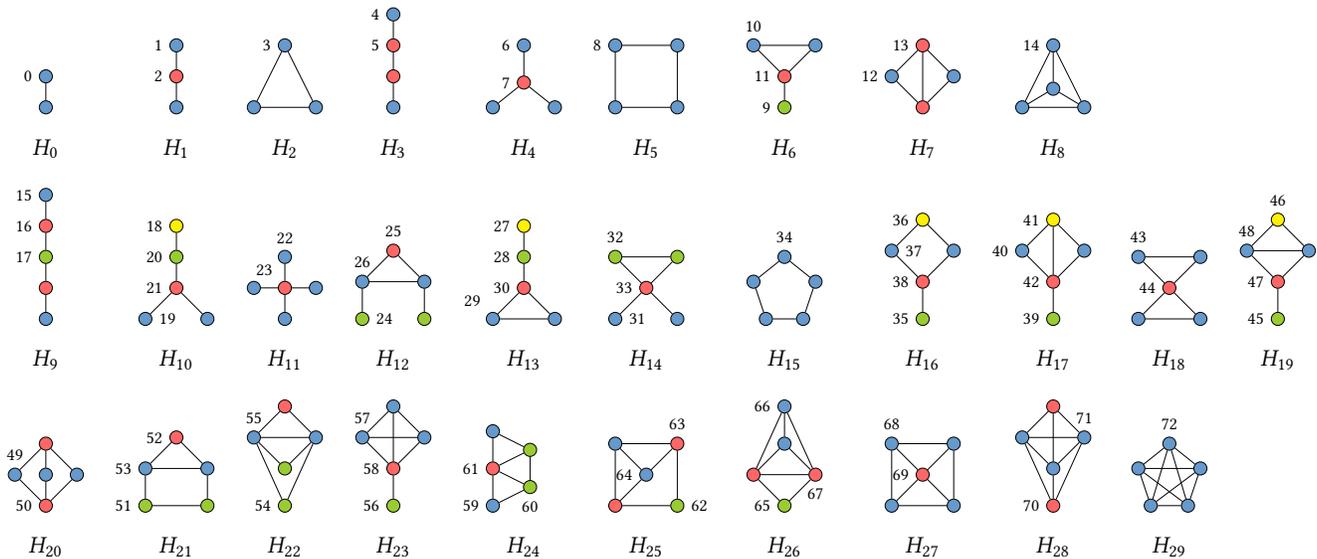
\begin{figure*}[t]
\centering
\resizebox{1\textwidth}{!}{
  \begin{tikzpicture}[nd/.style={scale=1,circle,draw,inner sep=2pt},minimum size = 6pt]    
    \matrix[column sep=0.4cm, row sep=0.4cm,ampersand replacement=\&]
    {
    \Hzero \&
    \Hone \&
    \Htwo \&
    \Hthree \&
    \Hfour \&
   	\Hfive \&
   	\Hsix \&
   	\Hseven \&
   	\Height \\
   	\Hnine \&
   	\Hten\&
   	\Heleven \&
   	\Htwelve \&
   	\Hthirteen \&
   	\Hfourteen \&
   	\Hfifteen \&
   	\Hsixteen \&
   	\Hseventeen \&
   	\Heighteen \&
   	\Hnineteen \\
   	\Htwenty \&
   	\Htwentyone \&
   	\Htwentytwo \&
   	\Htwentythree \&
   	\Htwentyfour \&
  	\Htwentyfive \&
  	\Htwentysix \&
  	\Htwentyseven \&
   	\Htwentyeight \&
   	\Htwentynine \&  	
   	\\
 };
  \end{tikzpicture}
}
\caption{\label{fig:node_orbits} \small All vertex orbits for 5-vertex patterns. Within
any pattern, vertices of the same color form an orbit.}
\end{figure*}

\subsection{Main Contributions} \label{sec:cont}

Our primary result is the Efficient Vertex Orbit pacKagE (EVOKE), an algorithm
to compute all 5-VOCs.

{\bf Practical local counting:} EVOKE advances the state of the art of subgraph counting.
It is the first algorithm that can feasibly obtain all 5-VOCs
on graphs with tens of millions of edges. We do comprehensive tests on many public
data sets. We observe that EVOKE gets counts on graphs with millions of edges
in just minutes, and on graphs with tens of millions of edges within an hour. This
is on a single commodity machine with 64GB memory, without any parallelization.
In contrast, for the larger instances, 
the previous state of the art ORCA package takes more than two days or runs out of memory,
on a more powerful machine (384GB RAM). Even on instances where ORCA terminates, EVOKE
is about a hundred times faster. We show the speedup of EVOKE over ORCA in \Fig{speedup}.
EVOKE is also able to get 5-VOCs in a social network with 100M edges, in less than two days.
(ORCA runs of out memory in such instances.)
All the blue bars in \Fig{speedup} denote instances where ORCA runs of out of memory (in two days) or
is a thousand times slower than EVOKE.

EVOKE has a large number of independent sub-algorithms. It is straightforward to
run them in parallel, and we get about a factor two speedup. We do not consider
this a significant novelty of EVOKE, but it does allow for an even faster running time.

{\bf Local counting without enumeration:} Our work builds on the ESCAPE framework
of Pinar-Seshadhri-Vishal~\cite{PiSeVi17}. One of their main insights
is a combination of graph orientations and a ``pattern cutting'' technique. Larger
patterns are carefully cut into smaller patterns. It is then shown that local counts
of \emph{smaller} patterns can be combined into \emph{global} (total) counts of larger patterns.
We formally prove that, for the orbits in \Fig{node_orbits}, one
can generalize their method to VOCs. This is mathematically quite technical
and requires manipulations of various pattern automorphisms (which is not required
for total counts). But the final result is a large collection
of polynomial formulas to compute individual VOCs through some specialized local counts
of smaller subgraphs. EVOKE exploits the structure among these formulas to count all VOCs
efficiently.

Somewhat surprisingly, we mathematically prove that the running time is only a constant factor
more than that of ESCAPE (which only computes total counts). This is borne out empirically
where the running time of EVOKE is typically twice that of ESCAPE. Our
result demonstrates the power of the cutting framework introduced in~\cite{PiSeVi17}.

{\bf Fast computation of orbit frequency distributions:} The distribution
of VOCs is a useful tool in graph analysis, often called \emph{graphlet
degree distribution} in bioinformatics~\cite{Pr07}. EVOKE makes it feasible to compute
these distributions over real data. As a small demonstration of EVOKE, we observe 
interesting behavior in VOCs across graphs from different domains. Also,
the VOCs of different orbits within the same pattern behave differently, showing
the importance of getting such fine-grained information.

{\bf On 4-VOCs:} We do not consider this as a new contribution, but
a salient observation for those interested in subgraph counting. EVOKE
determines all 4-VOCs as a preprocessing step, based on ideas
in~\cite{PiSeVi17} and Ortmann-Brandes~\cite{ortmann2017efficient}. As stated
in these results, the key insight is an implementation of an old algorithm of
Chiba-Nishizeki for 4-cycle counting~\cite{ChNi85}. This method is incredibly fast,
and computes 4-VOCs in minutes. (Even for the largest instance of
more than 100M edges, it took less than an hour.) For example, for a LiveJournal
social network with 42M edges, EVOKE took ten minutes on a commodity machine
(we got the same time even on a laptop). Contrast this
with previous results for counting 4-VOCs for the same graph,
which used a MapReduce cluster~\cite{ElShBo16}. (We note that EVOKE, and the other results,
are technically computing edge orbit counts, a more general problem.)

\subsection{Related Work} \label{sec:related}

Subgraph counting is an immensely rich area of study, and we refer the reader
to a tutorial for more details~\cite{SeTi19}. Here, we only document
results relevant to our problem. For this reason, we do not discuss the
extremely large body of work on triangle counting (the most basic subgraph counting problem).

Vertex orbit counts beyond triangles have found significant uses in network analysis and machine learning.
Notably, Shervashidze-Vishwanathan-Petri-Mehlhorn-Borgwardt defined the \emph{graphlet kernel}, that uses vertex orbits counts
to get embeddings of vertices in a network~\cite{ShervashidzeVPMB09}.
Ugander-Backstrom-Kleinberg showed that 
4-vertex orbit counts can be used for role discovery and distinguishing different
types of graph neighborhoods~\cite{UganderBK13}. In an exciting recent use of orbit counts,
Rotabi-Kamath-Kleinberg-Sharma showed that four and five cycle counts can be used for weak tie discovery
in the Twitter network~\cite{RKKS17}. Yin-Benson-Leskovec have defined higher-order
clustering coefficients, which are ratios of specific orbit counts~\cite{YiBiKe18,YiBiKe19}.
There is a line of work on
the surprising benefits of using cycle and clique counts as vertex or edge weights,
to find denser and more relevant communities in networks~\cite{BeHe+11,SaSe15,Ts15,BeGlLe16,TPM17}.

We now discuss the literature on algorithms for subgraph counting. Ahmed-Neville-Rossi-Duffield
gave the first algorithm that could count (total) 4-vertex subgraph counts for graphs with millions of edges~\cite{AhNe+15}. Their PGD package was a significant improvement
over past practical work for this problem~\cite{GonenS09}. Pinar-Seshadhri-Vishal designed
the ESCAPE algorithm for practical (total) 5-vertex subgraph counting~\cite{PiSeVi17}. While
these algorithms employed many clever combinatorial ideas, they did not focus on vertex orbit counting.
There was concurrent development of sampling algorithms that are orders
of magnitude faster, such as path-sampling~\cite{JhSePi15} and the MOSS package~\cite{WaZh+18}. 

Elenberg-Shanmugam-Borokhovich-Dimakis gave algorithms for 3, 4-vertex orbit counting~\cite{ElShBo15,ElShBo16}. They employed a randomized algorithm, and proved convergence through polynomial
concentration inequalities. The number of samples required for concentration was large,
and they used Map-Reduce clusters to process graphs with tens of millions of edges.
It was observed implicitly in the ESCAPE package and explicitly, 
by Ortmann-Brandes~\cite{ortmann2017efficient} that ideas from  
a classic result of Chiba-Nishizeki~\cite{ChNi85} gave a faster, exact algorithm
for 4-vertex orbits.

The state of the art for local counting of 5-vertex orbits is the ORCA package
of Ho\v{c}evar-Dem\v{s}ar~\cite{hovcevar2016computation}. The algorithm is based
on a method to build sets of linear equations relating various orbit counts. This saves
computing all orbit counts independently. With some careful choices, ORCA
tries to perform enumeration on the ``easier'' counts, and get the ``harder'' counts
through the linear equations. There were also results on generating
these linear equations auotmatically~\cite{MAM+16,MAM+18}. 
We note that ORCA also has algorithms to generate 5-edge orbit counts, but this takes
even longer than 5-VOCs. We leave the generalization of EVOKE to edge orbit counts as
future work.

Rossi-Ahmed-Carranza-Arbour-Rao-Kim-Koh
proposed a parallel algorithm for counting typed graphlets (subgraph patterns), which are a generalization of subgraph patterns to heterogeneous networks~\cite{rossi2019heterogeneous}.

Most of the exact subgraph and orbit counting algorithms work on subgraphs of size 5 or less. 
Algorithms that attempt to count subgraphs beyond that size typically use randomization,
which has inspired a rich literature~\cite{BKS02, JoGh05, WernickeRasche06,TsKaMiFa09,  KoMiPeTs10, YoKi11, PaTs12, jha2013space, ZhWaBu+12, BhRaRa+12,SePiKo13, KoPiPlSe13, KP13, PaTaTi+13, RaBhHa14, SePiKo14, JhSePi15, LimK15, MVV16, JaSe17, Shin17, CormodeJ17, WaZh+18}.

\section{Preliminaries} \label{sec:prelim}
The input is an undirected simple graph $G= (V, E)$, with $n$ vertices and $m$ edges.
The patterns of interest are all connected subgraphs with at most $5$ vertices,
denoted $H_0, \ldots, H_{29}$, as shown in \Fig{node_orbits}. 
Previous results in~\cite{PiSeVi17} show that disconnected pattern counts can be
determined by inclusion-exclusion from all connected pattern counts. Hence,
we only focus on connected pattern subgraphs.

We now formally define orbits.
The definitions below are taken from Bondy and Murty (Chapter 1, Section 2)~\cite{bondy2008graph}.

\begin{definition} \label{def:auto} Fix labeled graph $H = (V(H), E(H))$.
An \emph{automorphism} is a bijection $\sigma: V(H) \to V(H)$ such that
$(u,v) \in E(H)$ iff $(\sigma(u), \sigma(v)) \in E(H)$.

Define an equivalence relation among $V(H)$ as follows. We say
that $u \sim v$ ($u, v \in V(H)$) iff there exists an automorphism
that maps $u$ to $v$. The equivalence classes of the relation are
called \emph{orbits}.
\end{definition}

\Fig{node_orbits} shows the 73 different orbits. Within any $H_i$,
all vertices in an orbit are colored the same. For example, in $H_{28}$,
there are two different orbits (blue and red). The blue (resp. red) vertices
can be mapped to each other by automorphisms, and are therefore ``equivalent''.

Technically, we denote orbits as pairs $(H,S)$, where $H$ is a (labeled) pattern
subgraph and $S$ is the subset of vertices forming the orbit. 
Consider pattern $H$ and orbit $\theta = (H,S)$. We denote:

\begin{asparaitem}
    \item $\orb(H)$: The set of orbits in the pattern $H$.
    \item $\sz(\theta)$: $|S|$, the number of vertices in the orbit $\theta$.
\end{asparaitem}

\medskip

We can similarly define edge orbits as follows.

\begin{definition} \label{def:edge-orbit} Fix labeled graph $H = (V(H), E(H))$.
Define an equivalence relation among $E(H)$ as follows. Let $e=(u,v)$ and $e^\prime=(w,x)$ be two edges in $H$. We say that $e \sim e^\prime$ iff there exists an automorphism that maps $e$ to $e^\prime$ (i.e. it maps $u$ to $w$ and $v$ to $x$ or $u$ to $x$ and $v$ to $w$). The equivalence classes of the relation are called \emph{edge-orbits}.
\end{definition}

\medskip

{\bf Induced vs non-induced:} A non-induced subgraph is obtained by taking
a subset of edges. An induced subgraph is obtained by taking a subset of vertices
and considering all edges and non-edges among them. (A clique contains
all non-induced subgraphs of smaller sizes, but the only induced subgraphs
it contains are smaller cliques.) A theorem in~\cite{PiSeVi17} proves
that the vector of non-induced subgraph (up to a given size) counts
can be converted to the corresponding induced counts, through a linear
transformation. 
A directed generalization of the arguments
holds for $k$-VOCs, in that non-induced orbit counts (for each vertex)
can be converted to induced orbit counts by a linear transformation. For readability, we move the details to~\Sec{transform} of Appendix. It is a small linear
transformation of the 73-dimensional orbit count for each vertex, and is efficient to do on all vertices.

EVOKE computes both non-induced and induced counts. Algorithmically, it is easier
to compute non-induced counts first; hence we shall only refer to them in the technical
description.

We are ready to define VOCs.

\begin{definition} \label{def:unlabeled-match} Fix an orbit $\theta = (H,S)$ and a vertex $v \in V$ (in the input
graph $G$). A \emph{match} of $\theta$ involving $v$ is a non-induced copy of $H$ in $G$ such that $v$ is mapped to a vertex in $S$. Call two matches \emph{equivalent}, if one can be obtained from the other by applying an automorphism. We define $\uoMatchSize(v,\theta)$ to be the number
of distinct matches of $\theta$ involving $v$.
\end{definition}
Our aim is to compute the entire list of numbers $\{UM(v,\theta)\}$, over all $v \in V$
and all $\theta$ in \Fig{node_orbits}.

\medskip

{\bf Degree ordering:} We will use the \emph{degree orientation}, a fundamental tool for subgraph counting
that was pioneered by Chiba-Nishizeki~\cite{ChNi85}. 
We will convert $G$ into an DAG $G^\rightarrow$ as follows. 
Let $\prec$ denote the \emph{degree ordering} of $G$. For vertices $i,j$, we say $i \prec j$, if either $d(i) < d(j)$ or $d(i) = d(j)$ and $i < j$ (comparing vertex id). 
The DAG $G^\rightarrow$ is obtained by orienting the edges with respect to $\prec$ ordering. 
In both the algorithm and analysis, all references to directed structures are with
respect to $G^\rightarrow$.

{\bf Notation for subgraph counts:} In our formulas for orbit counts, we will use
the following notation. We use $d(v)$ for the degree of vertex $v$. We will
use $W(G)$, $\cc(G)$, $\dpath(G^\rightarrow)$, and $\dbp(G^\rightarrow)$ for 
the total count of wedges, diamonds, directed 3-paths, and directed
bipyramids respectively. These subgraphs are shown in~\Fig{5vertex-basis}.

\begin{figure}[t]
\centering
\resizebox{0.45\textwidth}{!}{
  \begin{tikzpicture}[nd/.style={scale=1,circle,draw,fill=bluegray,inner sep=2pt},minimum size = 6pt]  
    \matrix[column sep=0.4cm, row sep=0.2cm,ampersand replacement=\&]
    {
     \Wedge \&
     \Diamond \&
     \DirThreePath \&
     \DirBiPyramid \\
 };
  \end{tikzpicture}  
  }
\caption{\small Fundamental patterns enumerated for orbit counting}
\label{fig:5vertex-basis}
\end{figure}
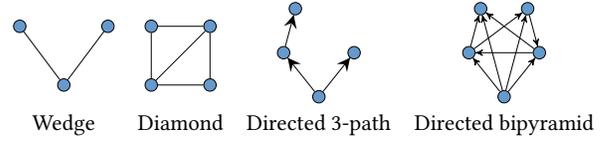

\subsection{Main theorem} \label{sec:mainthm}

\begin{theorem} \label{thm:running-time}
There is an algorithm for exactly counting all VOCs for orbits $0$-$72$, 
whose running time is $O(W(G) + \cc(G) + \dpath(G^\rightarrow) + \dbp(G^\rightarrow) + m + n)$.
\end{theorem}

This theorem is analogous to that of ESCAPE (\cite{PiSeVi17}) which gives
the same asymptotic running time for just total counting of 5-vertex subgraphs. 
We consider it quite significant that one gets the same asymptotic running time,
despite the output being much larger and far more fine-grained. We stress
that the EVOKE algorithm is significantly different than ESCAPE, since
the orbit counts behave differently from total subgraph counts.
The final proof is long, and is based on a collection of more than 50 equations for counting different orbits. So, we move the final proof and the equations to~\Sec{fivevertex} of the Appendix.

\section{Main ideas} \label{sec:ideas}

EVOKE builds off the ideas in ESCAPE for total subgraph counts. 
First, we explain difficulties in directly applying previous techniques.

{\bf Pattern cutting:} Intuitively, a 5-vertex pattern can be 
``cut'' into smaller
patterns that can be explicitly enumerated. An enumeration over these smaller
patterns can then be used to get a subgraph count. As an example, consider the 4-path ($H_9$).
By cutting at the center (green) vertex, one gets two wedges. Thus, we can 
basically square the number of wedges that end at a vertex, and then sum this
to get the total number of 4-paths. (Not quite, there is some inclusion-exclusion
required to ``correct'' this count, but it is fairly easy to work out.) But this 
fails for orbit counting. The 4-path has three distinct orbits, and the idea above
only works for the green orbit. 

This is even more problematic for patterns like $H_{21}$, $H_{25}$, $H_{27}$, $H_{28}$,
where the removal of certain vertices does not ``cut'' the pattern into convenient
smaller pieces. The main insight in ESCAPE was that all 5-vertex patterns have a convenient
cutset of vertices, whose removal leads to fragments that can be easily enumerated.
This is not true for orbits. We do have the freedom of choosing the convenient cutset.

{\bf From 4-edge orbit counts to 5-VOCs:} Our main insight is that the suitable generalization
of the pattern cutting approach connects 5-VOCs to 4-\emph{edge} orbit counts.
We essentially prove that many the orbit counts in \Fig{node_orbits} for
a vertex $v$ can be related (by non-trivial polynomial equations)
to the \emph{edge} orbits counts (of 4-vertex subgraphs) on edges incident
to $v$.
The edge orbits of 4-vertex subgraphs are given in~\Fig{edge_orbits}.
These edge orbits counts can be obtained by implementations of the Chiba-Nishizeki
clique and 4-cycle counter~\cite{ChNi85}, with extra inclusion-exclusion tricks to get all counts.
EVOKE uses this as a preprocessing step. We will give more details in~\Sec{counts}.

{\bf Careful indexing during enumeration:} Even with the previous ideas,
we still need an efficient implementation that can generate all the counts.
We design a collection of vertex and edge indexed
data structures, that are updated by an enumeration of the patterns
shown in \Fig{5vertex-basis}. Somewhat surprisingly, we show that as
these patterns are enumerated, one can quickly update these data
structures and generate all the orbit counts. This leads to \Thm{running-time}.

\section{The cutting framework for orbits} \label{sec:cut} 
In this section, we describe the cutting framework for orbits. As mentioned earlier,
this is a generalization of ideas in~\cite{PiSeVi17}. 

First, we formally define a match, which is a non-induced copy of $H$. For a set $C$ where $C \subseteq V(H)$, we use $H|_C$ to denote the subgraph of $H$ induced on $C$. We also denote the remaining graph after removing $C$ from $H$, by $H \setminus C$.

\begin{definition} \label{def:copy}
A \emph{match} of $H$ in $G$ is a bijection $\pi: T \rightarrow V(H)$ where $T \subset V$ and for any two vertices $t_1$ and $t_2$ in $T$, $(t_1, t_2)\in E$ if $(\pi(t_1), \pi(t_2)) \in E(H)$.
\end{definition}

\begin{definition} \label{def:orbit-match} 
Fix an orbit $\theta = (H,S)$ and a vertex $v \in V$. We define $\oMatch(v,\theta)$ to be the set of all (not necessarily distinct) matches $\pi: T \rightarrow V(H)$ of $H$, where $T \subset V$, such that $v \in T$ and $\pi(v) \in S$. We use $\oMatchSize(v, \theta)$ to denote $|\oMatch(v,\theta)|$. 
\end{definition}

\begin{definition} \label{def:representative}
For any orbit $\theta = (H,S)$ we define $\lambda = (H,i)$, where $i$ is a vertex in $S$, as a representative of $\theta$.\\
We use $r(\theta)$ to denote its representative $(H,j)$, where $j$ is the vertex with the smallest id in $S$.
\end{definition}

Let $\lambda = (H,i)$ be a representative of an orbit $\theta$. Abusing notation, for a vertex $v \in V$, we use $\oMatch(v,\lambda)$ to denote the set of matches $\pi \in \oMatch(v,\theta)$ where  $\pi(v) = i$. Analogously, we use $\oMatchSize(v, \lambda)$ to show $|\oMatch(v,\lambda)|$. We can see that $\oMatchSize(v, \theta) =  \sz(\theta) \cdot \oMatchSize(v, \lambda)$. Next, we define \emph{fragments} in $H$, which are the result of cutting $H$ using a cut set.

\begin{definition} \label{def:frag}
Let $H$ be a subgraph pattern and consider a non-trivial cut set $C \subsetneq V(H)$. Let $S_1, S_2, \ldots$ be connected components of $H\setminus C$. The \emph{fragments} of $H$ obtained by removing $C$ are the subgraphs of $H$ induced by $C \cup S_1, C \cup S_2, \ldots$. We denote the set of these fragments by $\frag_C(H)$.
\end{definition}

A partial match $\pi: T \rightarrow V(H)$ is similar to a match, except that it is an injection, and is not surjective, thus $|T| < |V(H)|$.

\begin{definition} \label{def:extention}
A match $\pi:T \rightarrow V(H)$ \emph{extends} a partial match $\sigma: T^\prime \rightarrow V(H)$ if $T^\prime \subset T$ and for any vertex $t$ in $T$, $\pi(t) = \sigma(t)$.
We denote the number of matches $\pi$ of $H$ that extend $\sigma$, by $\deg_H(\sigma)$.
\end{definition}

Consider a match $\sigma$ of $H|_C$. For $\sigma$ to extend to a match of $H$, it is sufficient that it extends to disjoint matches of all fragments in $\frag_C(H)$. Merging these extensions leads to a match of $H$. If extension of $\sigma$ to these fragments are not disjoint, merging them leads to a match of a different pattern $H^\prime$, which we call a \emph{shrinkage}.

\begin{definition} \label{def:shrinkage}
Let $H$, $H'$ be subgraph patterns, $C \subsetneq V(H)$ be a cut set of $H$, and $\frag_C(H)=\{F_1, F_2, \ldots, F_{|\frag_C(H)|}\}$. Let $\tau: H|_C \rightarrow H'$ be a partial match of $H^\prime$. For each $F_i \in \frag_C(H)$, let $\pi_i: F_i \rightarrow H'$ be a partial match of $H^\prime$ in $H$ that extends $\tau$. We call $\{\tau, \pi_1,$ $\pi_2, \ldots,$ $\pi_{|\frag_C(H)|}\}$ a $C$-shrinkage of $H$ into $H^\prime$ if for each edge $(s,t) \in E(H^\prime)$, there exists an edge $(a,b)$ in fragment $F_j \in \frag_C(H)$ such that $\pi_j(a)=s$ and $\pi_j(b)=t$.\\
We use $\shrink_C(H)$ to denote the set of patterns (up to isomorphism) $H^\prime$, to which there exist at least a $C$-shrinkage from $H$.
\end{definition}

\begin{definition} \label{def:number-shrinkage}
Consider graph $H$, $H' \in \shrink_C(H)$, $\lambda = (H,i)$, and $\lambda^\prime = (H^\prime,j)$. We define $\numshrink_C(\lambda, \lambda^\prime)$ to be the number of distinct $C$-shrinkages of $H$ into $H'$ where $\tau (i) = j$.
\end{definition}

\begin{lemma} \label{lem:cut-orb}
Consider a pettern $H$, an orbit $\theta =(H,S)$, a representative $\lambda = (H,i)$ of $\theta$, and a cut set $C$ in $H$ such that $i \in C$. Then,
\begin{align*}
\oMatchSize(v,\lambda) & = \sum_{\sigma \in \oMatch(v, (H|_C, i))} \prod_{F \in \frag_C(H)} \deg_F(\sigma) \\& - \sum_{H' \in \shrink_C(H)} \sum\limits_{\substack{\theta^\prime \in \orb(H^\prime), \\ \lambda^\prime = r(\theta^\prime)}} \numshrink_C(\lambda, \lambda^\prime) \cdot \uoMatchSize(v, \lambda^\prime)
\end{align*}
\end{lemma}

\begin{proof}
Consider any match $\sigma$ of $H|_C$ in $\oMatch(v, (H|_C, i))$, and all sets of maps $\{\pi_1,\ldots, \pi_{|\frag_C(H)|}\}$ where $\pi_\ell$ is a copy of $F_\ell \in \frag_C(H)$ that extends $\sigma$. The number of such sets is exactly:

\begin{equation} 
\sum_{\sigma \in \oMatch(v, (H|_C, i))} \prod_{F \in \frag_C(H)} \deg_F(\sigma) \label{eq:sum-over-maps}
\end{equation}

Consider one of these sets of maps $\{\pi_1,\ldots, \pi_{|\frag_C(H)|}\}$, let $V(\pi_\ell)$ be the set of vertices that $\pi_\ell$ maps to $F_\ell$. If all $V(\pi_\ell) \setminus V(C)$ are disjoint, we get a match in $\oMatch(v, \lambda)$.
Therefore, Each match of $H$ in $\oMatch(v, \lambda)$ is counted exactly one time in \Eqn{sum-over-maps}. But for each orbit $\theta^\prime = (H^\prime,S^\prime)$ where $H^\prime \in \shrink_C(H)$, we have also counted some matches in $\oMatch(v, \theta^\prime)$.
The number of distinct matches of $\theta^\prime$ involving $v$ is $\uoMatchSize(v,\theta^\prime)$. Let $\lambda^\prime=(H^\prime,j)$ be $r(\theta^\prime)$. The number of distinct $C$-shrinkages of $H$ into $H^\prime$, where $\tau(i) = j$, is $\numshrink_C(\lambda, \lambda^\prime)$. Thus, per each orbit $\theta^\prime$, we have counted $\numshrink_C(\lambda, \lambda^\prime) \cdot \uoMatchSize(v, \lambda^\prime)$ matches which should now be subtracted from \Eqn{sum-over-maps}.

The reason we considered only distinct matches of $\lambda^\prime$ involving $v$ is that the shrinkage from $H$ to $H^\prime$ gives us the labeling of $H^\prime$ and the set of maps $\{\pi_1,\ldots, \pi_{|\frag_C(H)|}\}$, which resulted in counting this match, dictates the match. Also, notice that the shrinkage determines the vertex in $H^\prime$ that $v$ is mapped to. That is why we consider number of shrinkages for a representative of $\theta^\prime$.
\end{proof}

\begin{corollary} \label{cor:orbit}
As mentioned, $\oMatchSize(v, \theta) =  \sz(\theta) \cdot \oMatchSize(v, \lambda)$. Therefore, we can derive $\uoMatchSize(v,\theta)$, which is the number of distinct matches of $\theta$, as follows:
$\uoMatchSize(v,\theta) = \sz(\theta) \cdot \oMatchSize(v,\lambda)/|Aut(H)|$.
\end{corollary}

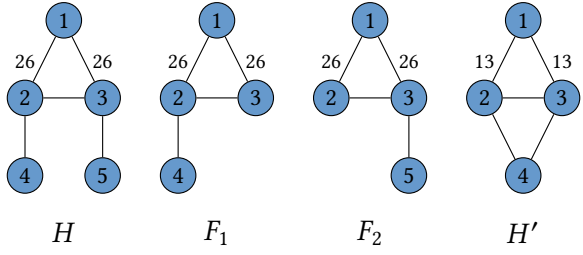
\begin{figure}[t]
\centering
\resizebox{.45\textwidth}{!}{
  \begin{tikzpicture}[nd/.style={scale=1, circle,draw,fill=bluegray,inner sep=2pt},minimum size = 13pt]    
    \matrix[column sep=0.5cm, row sep=0.2cm,ampersand replacement=\&]
    {
{
  \node (1) at (0.5,2) [nd] {1};
   \node (2) [label={\small 26}] at (0,1) [nd] {2};
   \node (3) [label={\small 26}] at (1,1) [nd] {3};
   \node (4) at (0,0) [nd] {4};
   \node (5) at (1,0) [nd] {5};
   \draw (1) to (2);
   \draw (1) to (3);
   \draw (2) to (3);
   \draw (2) to (4);
   \draw (3) to (5);
   
\node [inner sep = 0,below] at +(0.5,-0.5) {{\Large $H$}};
}     
     \& 

{
  \node (1) at (0.5,2) [nd] {1};
   \node (2) [label={\small 26}] at (0,1) [nd] {2};
   \node (3) [label={\small 26}] at (1,1) [nd] {3};
   \node (4) at (0,0) [nd] {4};
   \draw (1) to (2);
   \draw (1) to (3);
   \draw (2) to (3);
   \draw (2) to (4);
   
\node [inner sep = 0,below] at +(0.5,-0.5) {{\Large $F_1$}};
}     
     \& 

{
  \node (1) at (0.5,2) [nd] {1};
   \node (2) [label={\small 26}] at (0,1) [nd] {2};
   \node (3) [label={\small 26}] at (1,1) [nd] {3};
   \node (5) at (1,0) [nd] {5};
   \draw (1) to (2);
   \draw (1) to (3);
   \draw (2) to (3);
   \draw (3) to (5);
\node [inner sep = 0,below] at +(0.5,-0.5) {{\Large $F_2$}};
}     
     \& 
{
  \node (1) at (0.5,2) [nd] {1};
   \node (2) [label={\small 13}] at (0,1) [nd] {2};
   \node (3) [label={\small 13}] at (1,1) [nd] {3};
   \node (4) at (0.5,0) [nd] {4};
   \draw (1) to (2);
   \draw (1) to (3);
   \draw (2) to (3);
   \draw (2) to (4);
   \draw (3) to (4);
\node [inner sep = 0,below] at +(0.5,-0.5) {{\Large $H'$}};
}
 \\
 };
  \end{tikzpicture}  
  }
\caption{Application of \Lem{cut-orb} for vertex orbit 26}
\label{fig:orbit_26}
\end{figure}

{\bf Application of \Lem{cut-orb} for vertex orbit 26:}
We will show how this lemma works applying it to $H_{12}$ and computing VOCs for a vertex $v \in V$. Let $\theta_{26}=(H,S)$, where $S = \{2,3\}$ and $H$ is as shown in \Fig{orbit_26}, denote orbit 26. 
Let the representative $\lambda_{26}$ be $(H,2)$.

Let triangle $\{ 1,2,3 \}$ be the cut set $C$. So, $\frag_C(H) = \{F_1, F_2\}$ as we can see in \Fig{orbit_26}. Let $\hat{\lambda}=(H|_C,2)$ be a representative of orbit 3 (the only orbit in the cut set). Every triangle in $G$ incident to $v$ is a match in $\oMatch(v,\hat{\lambda})$. Each such triangle has two mappings to $H|_C$. consider triangle $\{u,v,w\}$ in $G$. Vertex $v$ has to be matched to vertex $2$, therefore one match ($A$) is $\sigma(u) = 1$, $\sigma(v) = 2$, and $\sigma(w) = 3$, and the other match ($B$) is $\sigma(u) = 3$, $\sigma(v) = 2$, and $\sigma(w) = 1$.
For match ($A$), $\deg_{F_1}(\sigma) \cdot \deg_{F_2}(\sigma) = (d(v)-2)(d(w)-2)$, and $\deg_{F_1}(\sigma) \cdot \deg_{F_2}(\sigma) = (d(v)-2)(d(u)-2)$ for match ($B$).

The only possible shrinkage of $H$ is to a diamond $H'$, as shown in \Fig{orbit_26}. Let orbit $\theta_{13}=(H^\prime,S^\prime)$, where $S^\prime=\{2,3\}$, show orbit 13. We can see that in any $C$-shrinkage of $H$ into $H'$, $\tau(2) \in S^\prime$. Let $\lambda_{13}=(H^\prime,2)$ be a representative of $\theta_{13}$. Notice that $\numshrink_C(\lambda_{26}, \lambda_{13}) = 2$.
In one case we set $\tau(1) = 1$, $\tau(2) = 2$, $\tau(3) = 3$, $\pi_1(4) = 4$, and $\pi_2(5) = 4$. In the other case, we set $\tau(1) = 4$, $\tau(2) = 2$, $\tau(3) = 3$, $\pi_1(4) = 1$, and $\pi_2(5) = 1$. The set of maps $\{\tau, \pi_1, \pi_2\}$ in both cases forms a $C$-shrinkage of $H$ into $H'$ where $\tau(2) = 2$.

\begin{align}
\oMatchSize(v,\lambda_{26}) & =  \sum_{t = \langle u,v,w\rangle \ \textrm{triangle}} [(d(v)-2) ((d(u)-2) \nonumber \\
&  +(d(w)-2))] - 2 \cdot  \uoMatchSize(v,\lambda_{13}) \label{eq:orb26}
\end{align}

Note that $\sz(\theta_{26}) = 2$ and $H$ has two automorphisms,
so (by \Cor{orbit}) $\uoMatchSize(v,\theta_{26}) = \oMatchSize(v,\lambda_{26})$ .

\section{Getting orbit counts} \label{sec:counts}
\Lem{cut-orb} gives us a collection
of more than fifty equations similar to \Eqn{orb26}. For each of them, we verify that they can be computed through an enumeration of the patterns in \Fig{5vertex-basis}, assuming that all edge orbits of \Fig{edge_orbits} are available. For readability, we move the details of equations for computing 5-VOC, their runtime analysis, and the final proof of \Thm{running-time} to \Sec{fivevertex} of the Appendix, but we give the details of 4 vertex and edge orbit counts and a few example of 5-VOCs equations in this section. We will prove the following theorem for 4-vertex orbit counting.

\begin{theorem} \label{thm:4orbit}
All vertex and edge orbit counts for 4-vertex patterns
can be obtained in time $O(W(G) + \cc(G) + m + n)$.
\end{theorem}

{\bf Getting 4-VOCs:} The easiest way to demonstrate our framework, is to apply it to orbits in 4-vertex patterns with up to 4 vertex (orbits 0-14). For each vertex $v$ in $G$, let $T(v)$, $C_4(v)$, and $K_4(v)$ denote the number of triangles incident to $v$, the number of 4-cycles incident to $v$, and the number of 4-cliques incident to $v$, respectively. For each edge $e=(u,v)$ in $G$, let $T(e)$, $C_4(e)$, and $K_4(e)$ denote the number of triangles incident to $e$, the number of 4-cyles incident to $e$, and the number of 4-cliques incident to $e$, respectively. For each triangle $t$, let $K_4(t)$ denote the number of 4-cliques including $t$. In~\cite{PiSeVi17}, Pinar-Seshadhri-Vishal have shown that there is an algorithm that in time $O(W(G) + \cc(G) + m + n)$, computes (for all vertices $u$, edges $e=(v,w)$, and triangles $t$): all $T(v)$, $T(e)$, $C_4(v)$, $C_4(e)$, $K_4(v)$, $K_4(e)$, and $K_4(t)$. Their algorithm also obtains for every edge $e$, the list of triangles incident to $e$. Vertex orbit counts of patterns with up to 4 vertices can be computed using the equations presented in \Lem{4vertexorbit}.

\begin{lemma} \label{lem:4vertexorbit}
For $i \in {0,\ldots,14}$, let $\lambda_i = r(\theta_i)$. Then, for each vertex $u \in V$,

$\uoMatchSize(u,\lambda_{0}) = d(u)$

$\uoMatchSize(u,\lambda_{1}) = \sum_{v \in N(u)} d(v) - 1$

$\uoMatchSize(u,\lambda_{2}) = \binom{d(u)}{2}$

$\uoMatchSize(u,\lambda_{3}) = T(u)$

$\uoMatchSize(u,\lambda_{4}) = \sum_{v \in N(u)} \uoMatchSize(v,\lambda_{1}) - 2\uoMatchSize(u,\lambda_{2}) -2\uoMatchSize(u,\lambda_{3})$

$\uoMatchSize(u,\lambda_{5}) = \uoMatchSize(u,\lambda_{1}) (d(u)-1) - 2\uoMatchSize(u,\lambda_{3})$

$\uoMatchSize(u,\lambda_{6}) = \sum_{v \in N(u)} \binom{d(v)-1}{2}$

$\uoMatchSize(u,\lambda_{7}) = \binom{d(u)}{3}$

$\uoMatchSize(u,\lambda_{8}) = C_4(u)$

$\uoMatchSize(u,\lambda_{9}) =  \sum_{v \in N(u)} \uoMatchSize(v,\lambda_{3}) - T(u,v)$

$\uoMatchSize(u,\lambda_{10}) = \sum_{v \in N(u)} T(u,v) (d(v) -2)$

$\uoMatchSize(u,\lambda_{11}) = \uoMatchSize(u,\lambda_{3}) (d(u) -2)$

$\uoMatchSize(u,\lambda_{12}) = \sum_{t=(u,v,x)} T(v,x)-1$

$\uoMatchSize(u,\lambda_{13}) = \sum_{v \in N(u)} \binom{T(u,v)}{2}$

$\uoMatchSize(u,\lambda_{14}) = K_4(u)$

\end{lemma}

\begin{proof}
For $\theta_{12}$, we use a triangle as the cut set. The remaining component is a vertex, which forms a triangle with the edge $(v,x)$. After mapping a triangle $t=(u,v,x)$ to the cut set, we need to select a vertex to extend the copy of the cut set to a copy of $\theta_{12}$. The number of such vertices are equal to $T(v,x)-1$, which is the number of all triangles incident to the edge $(v,x)$ except $t$.

The rest of the equations, either have a vertex or an edge as the cut set, or are computed directly, such as $\theta_2$, which are easy to follow.
\end{proof}

{\bf Getting edge orbit counts of 4-vertex subgraphs:} There are eleven
edge orbits for 4-vertex subgraphs as shown in \Fig{edge_orbits}. For an edge $(u,v)$, we use $E_i((u,v))$ to denote the count of the $i$th edge orbit (where $i$ is from \Fig{edge_orbits}).

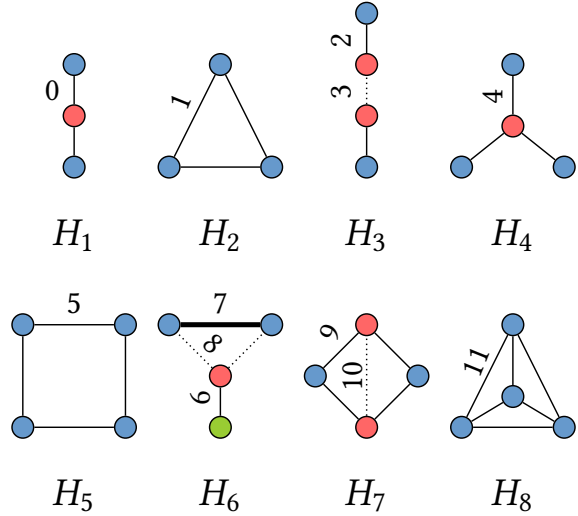
\begin{figure}[t]
\centering
\resizebox{0.45\textwidth}{!}{
\begin{tikzpicture}[nd/.style={scale=1,circle,draw,inner sep=2pt},minimum size = 6pt]    
    \matrix[column sep=0.2cm, row sep=0.4cm,ampersand replacement=\&]
    {
    \HoneEdge \&
    \HtwoEdge \&
    \HthreeEdge \&
    \HfourEdge \\
   	\HfiveEdge \&
   	\HsixEdge \&
   	\HsevenEdge \&
   	\HeightEdge \\
 };
  \end{tikzpicture}
}
\caption{\label{fig:edge_orbits} \small All edge orbits of 4-vertex patterns. Within each pattern, edges of the same line style form an edge orbit.}
\end{figure}

\begin{definition}
Given an edge $e=(v,u)$ in graph $G$ and edge orbit $i$ which lies in pattern $H$, a match of edge orbit $i$ involving edge $(v,u)$, is a non-induced copy of $H$ in $G$ such that $e$ is mapped to an edge in edge orbit $i$.

Let the vertex orbits of the two end points of edge orbit $i$, be $\theta_a=(H, S_a)$ and $\theta_b=(H,S_b)$ where $a \geq b$. From the definition of automorphism, it is clear that a match of edge orbit $i$, involving edge $e=(v,u)$, maps $v$ to a vertex in $S_a$ and $u$ to a vertex in $S_b$, or vice versa. Similar to vertex orbits, we call two matches of an edge orbits equivalent if one can be obtained from the other by applying an automorphism. We use $E_i(\langle v, u \rangle)$ to denote the number of distinct matches of edge orbit $i$ involving $e=(v,u)$, where $v$ is mapped to a vertex in $S_a$ and $u$ is mapped to a vertex in $S_b$. If $a=b$, then $E_i(\langle v, u \rangle)=E_i(\langle u, v \rangle)$, thus we use $E_i((v, u))$ to denote the number of distinct matches of orbit $i$ involving $e$.

\end{definition}

Edge orbit counts of patterns with up to 4 vertex can be computed using the equations presented in \Lem{4edgeorbit}.

\begin{lemma}\label{lem:4edgeorbit}
Let $\lambda_3 = r(\theta_3)$. For each edge $(u,v) \in E(G)$,

$E_{0}(\langle u,v \rangle) = d(u) - 1$

$E_{1}((u,v))  = T(u,v)$

$E_{2}(\langle u,v \rangle)  = \sum_{x \in N(u) \setminus v} \left[d(x)-1 \right] - E_1(u,v)$

$E_{3}((u,v)) = (d(u)-1)(d(v)-1)-E_1(u,v)$

$E_{4}(\langle u,v \rangle) = \binom{d(u)-1}{2}$

$E_{5}((u,v)) = C_4(u,v)$

$E_{6}(\langle u,v \rangle) = \uoMatchSize(u, \lambda_{3}) - E_1(u,v)$

$E_{7}((u,v)) = \sum_{t=(u,v,x)} d(x) -2$

$E_{8}(\langle u,v \rangle) = E_1(u,v)(d(u) - 2)$

$E_{9}(\langle u,v \rangle) = \sum_{t=(u,v,x)} E_1(u,x)-1$

$E_{10}((u,v)) = \binom{T(u,v)}{2}$

$E_{11}((u,v)) = K_4(u,v)$
\end{lemma}

\begin{proof}
For $E_7$ and $E_9$, we need to enumerate the triangles incident to $(u,v)$. This could be obtained by the algorithm presented in~\cite{PiSeVi17} in time $O(W(G)+m+n)$ for all edges. The rest of the edge orbits are either computed directly or have a vertex or an edge as a cut set and are easy to follow.
\end{proof}

Finally, we can prove \Thm{4orbit}.

\begin{proof}[Proof of Theorem \Thm{4orbit}]
All vertex and edge orbits of patterns with up to 4-vertices could be obtained from equations in \Lem{4vertexorbit} and \Lem{4edgeorbit}. For all vertices $v$, all edges $e$, and all triangles $t$, we can get $T(v)$, $T(e)$, $C_4(v)$, $C_4(e)$, $K_4(v)$, $K_4(e)$, $K_4(t)$, and also for all edges $e$ we can obtain the list of triangles incident to $e$ in $O(W(G) + \cc(G) + m + n)$~\cite{PiSeVi17}. Assuming we have these counts, the rest of the vertex and edge orbit counts are either computed directly, or use a vertex, edge, or a triangle a cut set. Therefore, we can obtain all the other orbit counts for 4-vertex patterns in $O(W(G) + m + n)$ extra time. Overall, it takes $O(W(G) + \cc(G) + m + n)$ time to get all vertex and edge orbit counts of 4-vertex  patterns.
\end{proof}

{\bf Getting 5-VOCs:} We demonstrate the main ideas through a number of examples.
\begin{asparaitem}
    \item Orbit 26: The pattern cutting framework gives \Eqn{orb26}. We can precompute
    and store degrees at all vertices. During an enumeration of all triangles, 
    one can compute the summand for each triangle. The triangles can be enumerated
    in $O(W(G))$ time (indeed, it can be done even faster using orientations).
    Orbit 13 belongs to a 4-vertex pattern, so $\uoMatchSize(2,\theta_{13})$
    is obtained from \Thm{4orbit}.
    \item Orbit 37: let $\lambda_{37}=r(\theta_{37})$ and $\lambda_{12}=r(\theta_{12})$, then
\begin{align}
\uoMatchSize(u, \lambda_{37}) = \sum_{v \in N(u)} \left[ E_5((u,v))(d(v) - 2) \right] - 2\uoMatchSize(u,\lambda_{12}) \label{eq:orb37}.
\end{align}
    After storing $E_5$-values on each edge, one can get this VOC by
    a triangle enumeration. Orbit 12 belongs to a 4-vertex pattern.
    \item Orbit 68:
\begin{align}
\uoMatchSize(u, \theta_{68}) & = \sum_{\substack{v,w \text{ where} \\ \langle u,v,w \rangle \text{ is a wedge}}} {D(u,v,w) \choose 2} \label{eq:orb68}
\end{align}
    This is a challenging orbit to count. The value $D(u,v,w)$ is the number of diamonds ($H_{7}$)
    that involves the vertices $u,v,w$. It is too expensive to precompute and store all these values,
    but we can do it piecemeal. With knowledge of triangles, we can enumerate all diamonds
    involving a fixed vertex $u$. This can be used to find all the relevant values. Overall,
    the total time is a diamond enumeration and a wedge enumeration.
\end{asparaitem}

Overall, this technique can analogously handle all orbits, barring the 5-cycle and 5-clique (each
of which as a single orbit). The 5-clique can be directly enumerated in time
$O(\dbp(G^\rightarrow))$, a consequence of the classic Chiba-Nishizeki algorithm~\cite{ChNi85}
and explicitly proven in~\cite{PiSeVi17}.

{\bf Dealing with 5-cycles:} This is a special case, and handled in the following theorem.
This is a significant strengthening of 5-cycle counter in ESCAPE, which only gave a global
count in the same running time.

\begin{theorem} \label{thm:5-Cycle-Clique_running-time}
Vertex orbit counts for the 5-cycle can be computed in time $O(W(G) + \dpath(G^\rightarrow) + m + n)$.
\end{theorem}

\begin{proof} As shown in \Fig{5cycle}, there are three different 5-cycles DAGs up to isomorphism. 
Each 5-cycle has exactly one directed 3-path as shown in~\Fig{5cycle}, such that the remaining wedge is not an in-in wedge. In the figure, this directed 3-path is labeled $i, j, k, l$, and $w$ is the center vertex of the wedge. By a directed wedge enumeration,
we can precompute the number of such wedges between all pairs of vertices.
We enumerate over the directed 3-paths: for every directed 3-path we get between vertices
$i$ and $l$, we already know the number of relevant directed wedges between $i$ and $l$ as shown in~\Fig{5cycle}. This allows us to increment the orbit counts for the vertices $i,j,k,l$, by the number of wedges. Notice that an edge between $i$ and $k$ or between $j$ and $l$ could result in such a directed wedge. We can check the existence of these two edges using hashed edges of $G$. For each such edge, we should decrement the orbit counts for the vertices $i,j,k,l$ by one.

This process does not update the orbit count for vertex $w$. Let $P(i,l)$ be the
number of directed 3-paths from $i$ to $l$ as shown in~\Fig{5cycle}. To compute the orbit counts for vertex $w$, we enumerate in-out and out-out wedges between $i$ and $l$, and add $P(i,l)$ to the 5-cycle orbit count of vertex $w$. Notice that the 3-paths (corresponding to $P(i,l)$) potentially intersect with the wedge under consideration. Any such intersection would result in couting a tailed triangle instead of a 5-cycle. We need to subtract out the count of these tailed-triangles.
Any in-out wedge from $i$ to $l$ corresponds to a 5-cycle of type (c), in which case we count correctly. An in-out wedge from $l$ to $i$ corresponds to a 5-cycle of type (b); in this case we will count each tailed triangle of type (1), shown in \Fig{TaileTri}, as a 5-cycle while passing over $(l,j,i)$ wedge. An out-out wedge between $i$ and $l$ corresponds to a 5-cycle type (a); in this case we will count each tailed triangle of type (2) as a 5-cycle while passing over $(i,j,l)$, and count each tailed triangle of type (3) while passing over $(i,k,l)$.

We can easily get the tailed triangle counts corresponding to each wedge using the per-edge tailed triangle counts that we already have. All in all, we can get VOCs for the 5-cycle in the stated time.
\end{proof}

\begin{figure}[t]
\centering
\resizebox{.45\textwidth}{!}{
\begin{tikzpicture}[nd/.style={scale=1,circle,draw,fill=bluegray,inner sep=2pt},minimum size = 13pt]    
    \matrix[column sep=0.3cm, row sep=0.4cm,ampersand replacement=\&]
    {
     \DirFiveCycleOne \&
     \DirFiveCycleTwo \&
     \DirFiveCycleThree \&
     \DirThreePathForFiveCycle \\
 };
  \end{tikzpicture}
  }
\caption{\small All different 5-cycle DAGs up to isomorphism. There is only one directed 3-path as shown on the right side in each 5-cycle DAG where the remaining wedge is not an in-in wedge.}
\label{fig:5cycle}
\end{figure}
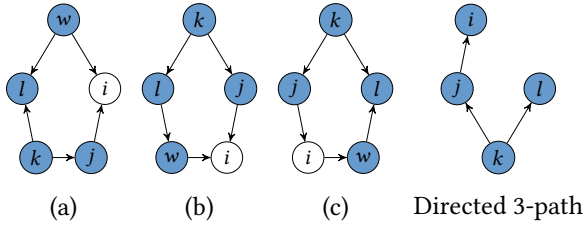

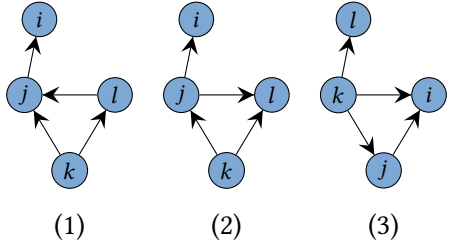
\begin{figure}[t]
\centering
\resizebox{.35\textwidth}{!}{
      \begin{tikzpicture}[nd/.style={scale=1,circle,draw,fill=bluegray!85,inner sep=2pt},minimum size = 13pt]    
    \matrix[column sep=0.4cm, row sep=0.2cm,ampersand replacement=\&]
    {
     \DirTailedTriOne \&
     \DirTailedTriTwo \&
     \DirTailedTriThree \\
 };
  \end{tikzpicture}
  }
\caption{Directed tailed triangles counted while counting 5-cycles}
\label{fig:TaileTri}
\end{figure}

\section{Experimental Results} 
We implement EVOKE in {\tt C++}.
We ran experiments on a commodity machine from AWS EC2: R5d.2xlarge, which has
Intel Xeon Platinum 8175M CPU @ 2.50GHz with 4 cores and 1024K L2 cache (per core), 34MB L3 cache, and 64GB memory.
For running EVOKE on the {\tt com-orkut} graph (117M edges), 
we used the more powerful R5d.12xlarge EC2 instance (with 384GB RAM).
We actually run ORCA for 5-vertex patterns on the larger machine for any instance
with more than 1M edges. 
The EVOKE package is available at~\cite{evoc} as open source code.

We used large graph datasets from the Network Repository~\cite{nr}, SNAP~\cite{snapnets}, and Citation Network Dataset~\cite{CitationNetwork, tang2008arnetminer}. We removed directions from edges, and omitted duplicates and self loops. \Tab{properties} includes the number of nodes, edges, and triangles for all the graphs we used. We also run EVOKE on {\tt wiki-en-cat}, a bipartite graph from the KONECT network repository~\cite{konect:2016:wiki-en-cat, download.wikimedia.org, konect}.

As mentioned earlier, we compare our results with ORCA~\cite{hovcevar2016computation} which is the state of the art algorithm for computing all 5-VOCs. The runtimes of ESCAPE, EVOKE, and ORCA is given in \Tab{properties}. We also state the time for just counting 4-VOCs. When we do not report
a time for ORCA, it implies that either ORCA ran out of memory or ran more than 1000 times
the EVOKE running time.
In all the results, the time includes the I/O, so we account for the time required
to print the (large) output into files.
As mentioned later, there is a parallel implementation of EVOKE, but all run times reported are of the sequential implementation
(to have a fair comparison with ORCA).

{\bf Running time of EVOKE:} As seen in \Tab{properties}, for many instances
of counting 5-VOCs, we simply cannot get results with ORCA. 
For all graphs larger than {\tt web-google-dir}, ORCA-5 runs out of memory even on the more powerful EC2 instance, or was stopped after a thousand times the corresponding EVOKE running time has passed (shown by blue bars in \Fig{speedup}).
When ORCA does give results, the speedup of EVOKE is easily in the orders of hundreds. \Fig{speedup}
gives the speedup as a chart. EVOKE makes 5-VOCs computation feasible, for graphs
with tens of millions of edges. ORCA is unable to process any graph in that size
range. Even for the large {\tt com-orkut} graph
with over 100M edges, EVOKE gets all counts in two days. 

As an aside, for counting 4-VOCs, EVOKE runs
typically in minutes, consistent with previous work~\cite{PiSeVi17,ortmann2017efficient}.

\begin{table*}[t]
\scriptsize
\caption{\small Properties of the graphs and runtime of ESCAPE, EVOKE, and ORCA} 
\begin{center}
\begin{tabular}{l|rrr|rrrrrr|}
\multicolumn{1}{c|}{} & \multicolumn{3}{c|}{} & \multicolumn{6}{c|}{Runtimes in seconds}\\
Dataset (sorted by increasing $|E|$) & 	$|V|$	& $|E|$	&  $|T|$ & ESC-4 & EVOKE-4 & ORCA-4 & ESC-5 & EVOKE-5 & ORCA-5  \\\hline 
soc-brightkite & 56.7K &  213K	& 494K & 0.43 & 0.59 & 1.77  & 4.69 & 7.74  & 562.84 \\
ia-email-EU-dir & 265K & 364K & 267K & 0.49 & 1.29 & 9.38 & 5.91 & 13.18  & 17.36K\\
tech-RL-caida & 191K & 607K & 455K & 0.68 & 1.29 & 2.99 & 4.65 & 10.03 & 595.44\\
Citation-network V1 & 2.17K & 631K & 248K & 0.69 & 2.57 & 42.91 & 2.89 & 8.93 & 275.15\\
ca-coauthors-dblp & 540K & 1.52M & 444M & 266.81 & 287.89 & 510.77 & 20.69K & 26.91K & 171.32K \\
DBLP-Citation-network V5 & 470K & 2.08M & 1.38M & 2.59 & 10.18 & 13.04 & 19.17 & 40.76 & 2.92K\\
Citation-network V2 & 660K & 3.02M & 1.9M & 4.11 & 11.57 & 28.42 & 32.78 & 69.36 & 7.52K\\
wiki-en-cat & 2.04M & 3.8M & 0 & 3.13 & 12.61 & 114.31 & 22.85 & 86.58 & - \\ 
web-google-dir & 876K & 4.32M & 13.4M & 4.76 & 10.03 & 45.40 & 45.86 & 104.88 & 76.37K\\	
web-wiki-ch-internal & 1.93M & 8.95M & 18.19M & 30.11 & 65.45 & 655.15 & 1.22K & 1.87K & - \\
tech-as-skitter & 1.69M & 11.1M & 28.8M & 28.91 & 68.25 & 827.46 & 853.21 & 1.46K & -\\
web-hudong & 1.98M & 14.43M & 21.61M & 48.20 & 85.83 & 1.78K & 2.41K & 3.45K & - \\
web-baidu-baike & 2.14M & 17.01M & 25.2M & 61.4 & 148.11 & 2.92K & 2.66K & 4.27K & - \\
tech-ip & 2.25M & 21.64M & 2.3M & 92.03 & 277.87 & 79.96K & 18.14K & 40.57K & - \\
soc-LiveJournal1  & 4.85M & 42.85M & 285.73M & 401.07 & 599.43 & 1.30K & 28.46K & 36.57K & -\\
com-orkut & 3.72M & 117.18M & 627.58M & 1.23K & 2.77K & 7.37K & 137.73K & 143.41K & -\\
\hline
\end{tabular}
\end{center}
\label{tab:properties}
\end{table*}

{\bf Comparison with ESCAPE:} 
\Thm{running-time} shows that the asymptotic upper bound given for ESCAPE in ~\cite{PiSeVi17} is also an asymptotic upper bound for EVOKE run time. 
We are able to validate this in practice. 
\Fig{EVOC_ESC_ratio} shows the ratio of runtime of EVOKE over ESCAPE for 5-vertex patterns. Note that ESCAPE counts subgraphs and EVOKE computes orbit counts for orbits in those subgraphs. 
As we can see in \Fig{EVOC_ESC_ratio}, in all our experiments the ratio is typically below 2 and never
more than 4. We believe this finding to be significant, since obtaining the richer
information of 5-VOCs is just as feasible as getting exact total counts.

{\bf Runtime distribution and parallel speedup:} Typically, a few orbits take the lion's share
of the running time.  \Fig{5OrbitTimeRatio} shows the split-up of running time over the various
orbits. We group them into four classes: the 5-clique, the 5-cycle, the orbits of $H_{25}$ and $H_{27}$
(these require diamond enumerations), and everything else. By and large,
just the 5-cycle and 5-clique orbits account for half the time. 

It is straightforward to parallelize the computation of these different groups
of orbits. For the non-induced setting, these are simply independent computations. We perform
this parallelism, and present the speedup in \Fig{parallel_speedup}. As expected,
there is roughly a 1.5-2 factor speedup, corresponding to the most expensive
orbit to compute.

\begin{figure*}[t]
\centering
\begin{subfigure}[b]{0.33\textwidth}
\centering
\includegraphics[width=\textwidth]{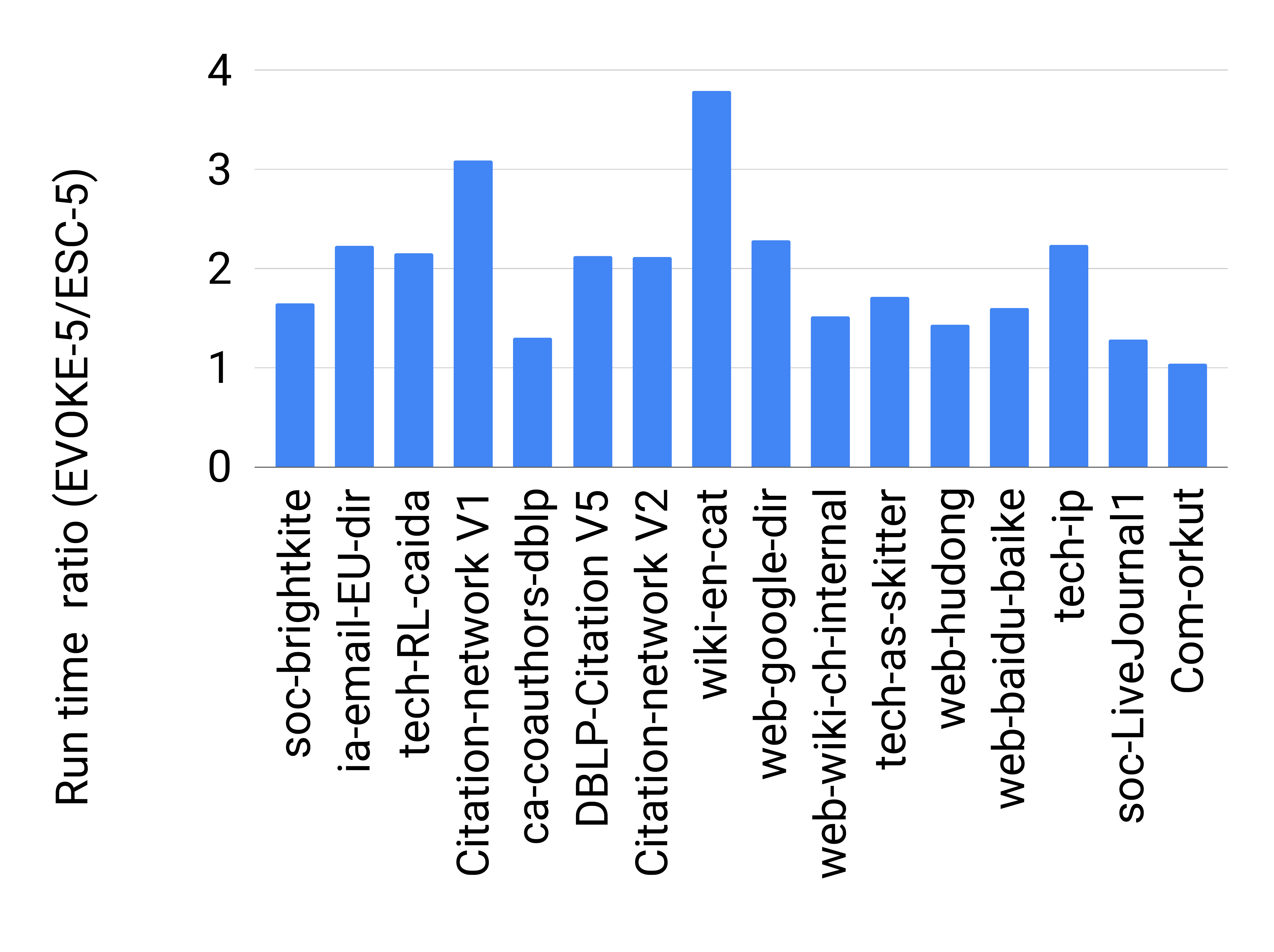}
\caption{Ratio of runtime represented as EVOKE-5/ESC-5 demonstrates \Thm{running-time}}
\label{fig:EVOC_ESC_ratio}
\end{subfigure}
\begin{subfigure}[b]{0.33\textwidth}
\centering
\includegraphics[width=\textwidth]{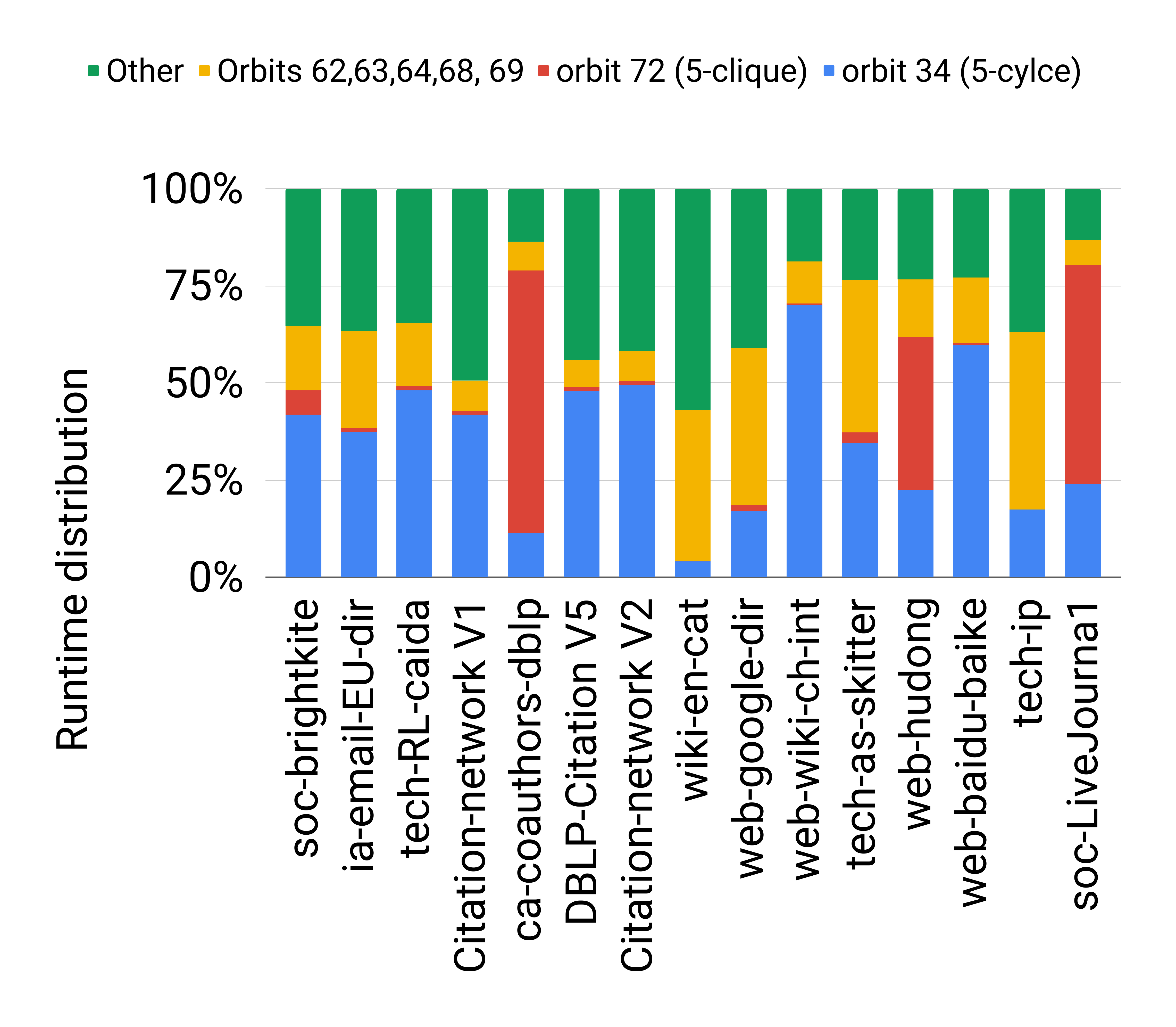}
\caption{Runtime distribution over 5-vertex orbits}
\label{fig:5OrbitTimeRatio}
\end{subfigure}
\begin{subfigure}[b]{0.33\textwidth}
\centering
\includegraphics[width=\textwidth]{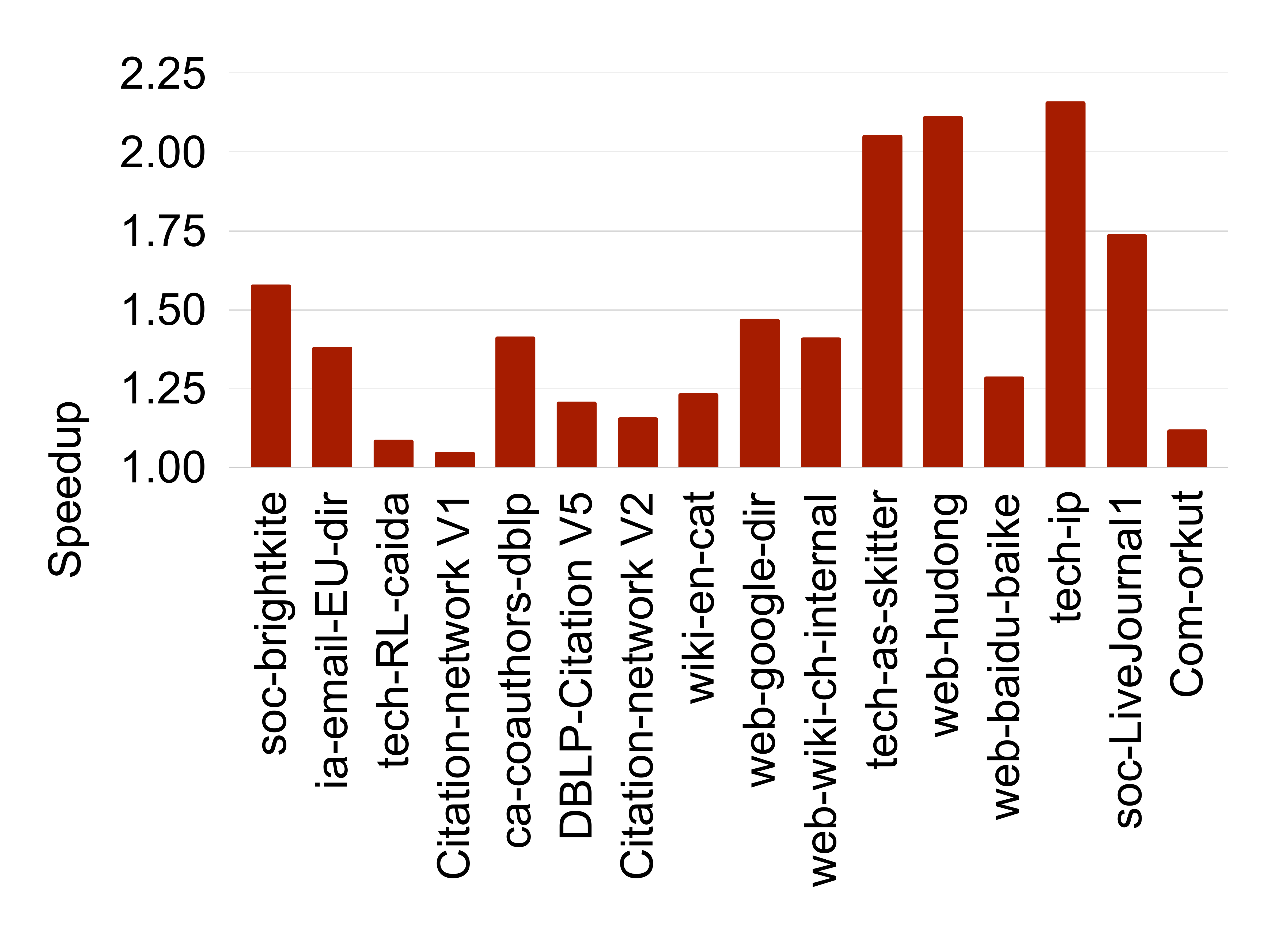}
\caption{Speedup achieved by parallel computation of 5-VOCs}
\label{fig:parallel_speedup}
\end{subfigure}
\caption{Empirical analysis of EVOKE runtime}
\label{fig:runtime_figs}
\end{figure*}

\begin{figure*}[h!]
\centering
\begin{subfigure}[b]{0.33\textwidth}
\centering
\includegraphics[width=\textwidth]{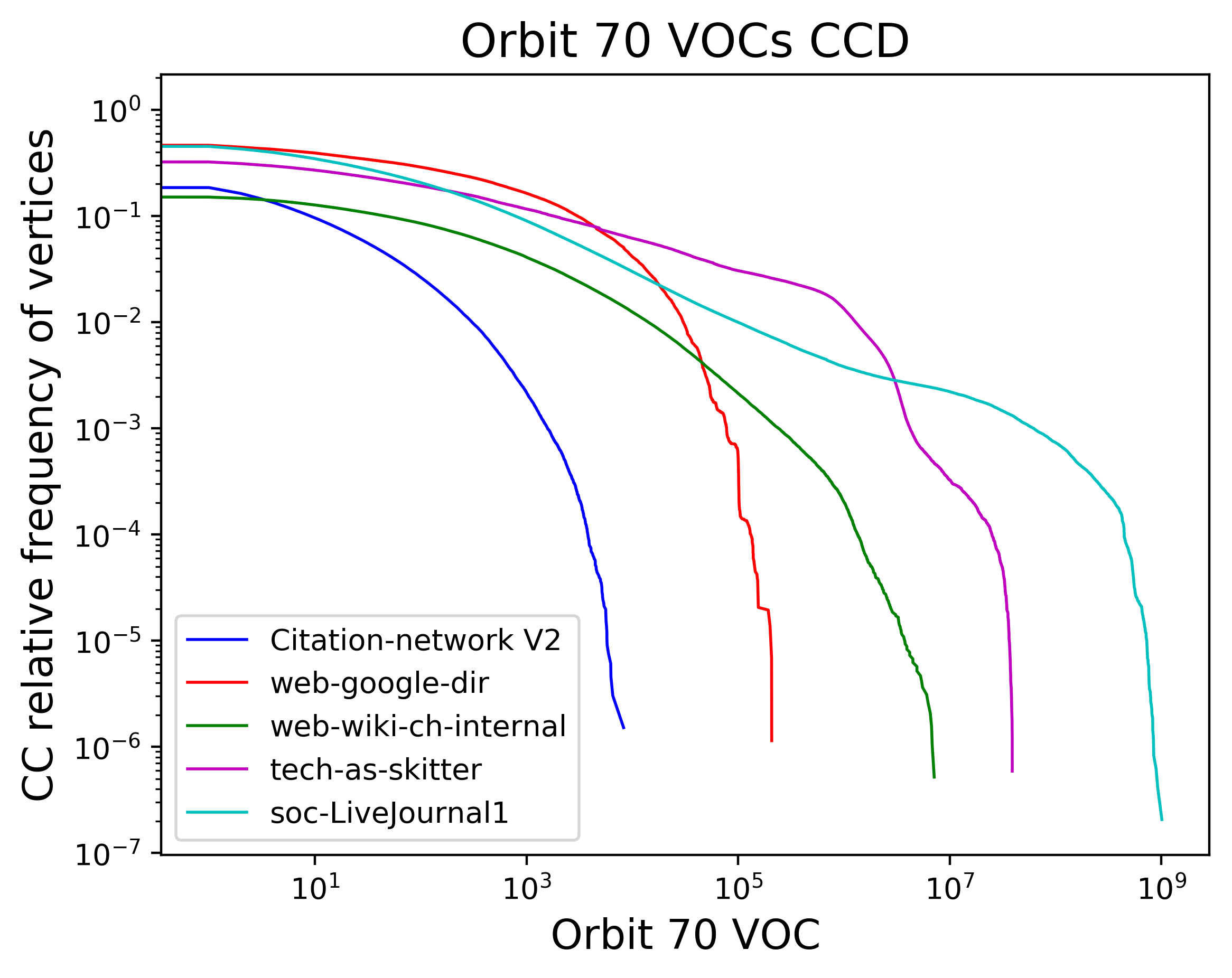}
\caption{Orbit 70 VOCs CCD}
\label{fig:orb70CCD}
\end{subfigure}
\begin{subfigure}[b]{0.33\textwidth}
\centering
\includegraphics[width=\textwidth]{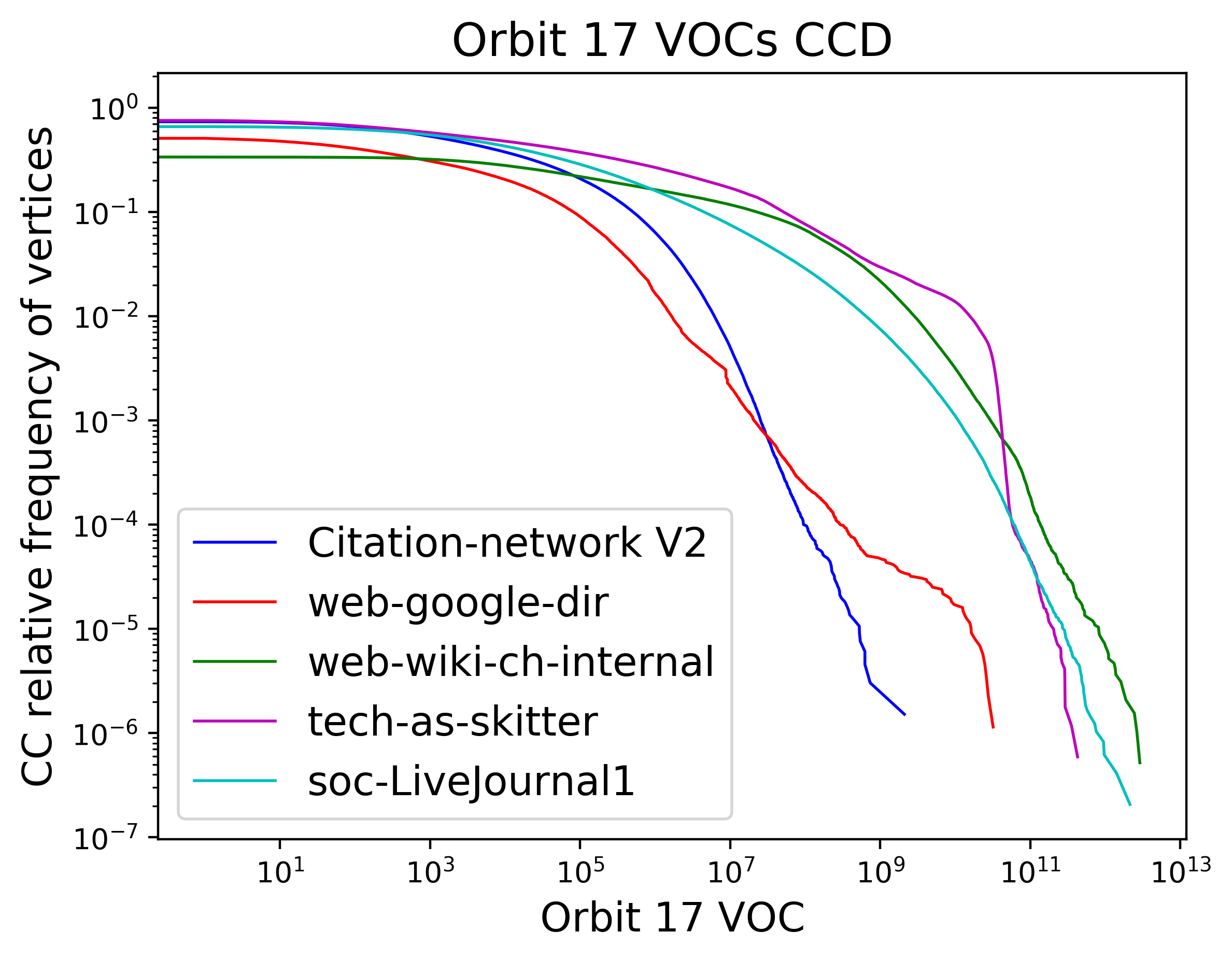}
\caption{Orbit 17 VOCs CCD}
\label{fig:orb17CCD}
\end{subfigure}
\begin{subfigure}[b]{0.33\textwidth}
\centering
\includegraphics[width=\textwidth]{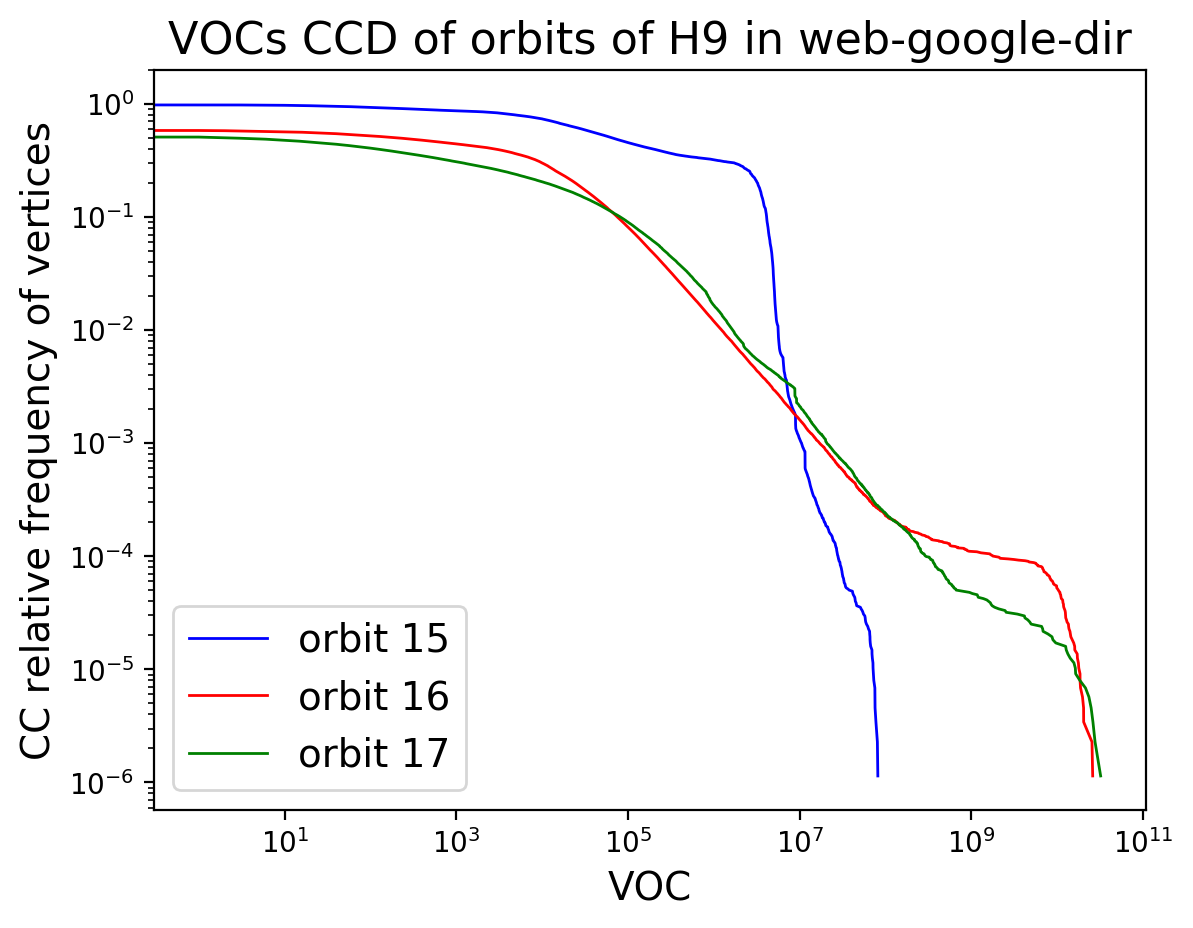}
\caption{VOCs CCD of $\orb(H_9)$}
\label{fig:H9CCD}
\end{subfigure}
\caption{(a), (b): VOCs comp. cum. distribution (CCD) of orbits. 
For count $x$, we plot the fraction of vertices with orbit count at least $x$. (c) 
For {\tt web-google-dir}, we plot the VOC CCD for all orbits of the 4-path. Observe that the distributions for the start/end (orbit 15) and the center (orbit 17) behave differently.}
\label{fig:count_dist_comp}
\end{figure*}

{\bf VOC distributions:} As a demonstration of EVOKE, we plot the VOC distribution
(also called graphlet degree distribution) of various graphs. To get cleaner figures,
we plot the Complementary Cumulative Distribution (CCD): for $x$, we plot the fraction
of vertices whose orbit count is at least $x$. This is plotted for Orbit 70 (in induced 5-clique minus edge) in~\Fig{orb70CCD} and for Orbit 17 (center of induced 4-path) in \Fig{orb17CCD}. We stress that these induced counts are typically
harder to obtain than the non-induced counts.

For Orbit 17, we observe that the largest count is more than trillions, showing the
challenges in exact counting. Also the distribution of {\tt tech-as-skitter} has a bigger dropoff
in the tail, which may be indicative of the path structures in AS networks.
The {\tt web-google-dir} graph has a sharp dropoff at the end as well.
We see that Orbit 70 distributions are quite different over the graphs,
unlike Orbit 17, where the tails are similar for three of the graphs. 
The counts in {\tt Citation-network V2} are much smaller, suggesting there are not many 5-cliques missing edges. 

In \Fig{H9CCD}, for the graph {\tt web-google-dir}, we plot the VOC
of the three different orbits (15-17) of the induced 4-path. Observe
how the distribution for Orbit 15 (the start/end) is significantly different from Orbit 17 (the center),
underscoring the fine-grained information that orbits provide over vanilla counts.

{\bf Graph mining through orbit counts:} As another demonstration, we focus
on the citation network {\tt DBLP-Citation-network V5}, where we have metadata associated with vertices
(papers). We found that the paper with largest count of Orbit 17 (center of induced 4-path)
is the classic book ``C4.5: Programs for Machine Learning'' by Ross Quinlan.
On the other hand, the paper participating in the most 5-cliques is the highly cited
VLDB 94 paper ``Fast Algorithms for Mining Association Rules in Large Databases'' by
Agarwal and Srikant. It is interesting that the orbit counts can immediately give
us semantically significant vertices.

\begin{acks}
We thank Akul Goyal for helpful discussions and his help on verifying the correctness of the output of the EVOKE package.
\end{acks}


\bibliographystyle{acm}
\bibliography{orbit_count}

\appendixpage
\appendix

\section{Conversion between Induced and Noninduced counts}
\label{sec:transform}

Given  $u \in V(G)$ and orbit $\theta$, we used $\uoMatchSize_G(u,\theta)$ to denote the number of distinct non-induced matches of $\theta$ in $G$. Let $\uoIMatchSize_G(u,\theta)$ denote the number of distinct induced orbit counts. Let $\theta_i=(H,S)$, where $H$ has $k$ vertices, where $k \leq 5$. It is easy to see that we can obtain $\uoMatchSize_G(u,\theta_i)$ from the set $\{\uoIMatchSize_G(u,\theta_i),\ldots,\uoIMatchSize_G(u,\theta_j)\}$, where $\theta_j$ is the vertex orbit with the largest index in $k$-vertex patterns.

Consider orbit $\theta_j=(H^\prime, S^\prime)$, where $j \geq i$. Let $v$ be a vertex in $H^\prime$, where $v\in S^\prime$. If we use graph $H^\prime$ as our input graph (instead of $G$), $\uoMatchSize_{H^\prime}(v,\theta_i)$ is actually the number of non-induced matches of $\theta_j$ in $G$ that an induced match of $\theta_j$ include as a subgraph.

If we think of the list of induced and non induced node orbit counts for any vertex $u$ in any graph $G$ as vectors $\uoIMatchSize(u)$ and $\uoMatchSize(u)$, there is a matrix $\bA$ such that $\uoMatchSize(u) = \bA \ \uoIMatchSize(u)$. The matrix for orbits which lie in 4-vertex patterns (orbits $\theta_4$-$\theta_{14}$) is given in \Fig{A}. The matrix for orbits of 5-vertex patterns (orbits $\theta_{15}$-$\theta_{72}$) is too large to be included here, but we made it accessible at~\cite{evoc}. Note that $\bA_{i,j}$ is the number of non-induced matches of $\theta_i$ that an induced match of orbit $\theta_j$ include, for any vertex $u$ in any graph $G$.

Naturally, $\uoIMatchSize(u) = \bA^{-1} \uoMatchSize(u)$, and that is how we get induced counts from non-induced counts. The inverse matrix for vertex orbits $\theta_4$-$\theta_{14}$ is given in~\Fig{Ainv}, the inverse matrix for orbits $\theta_{15}$-$\theta_{72}$ is again too large to be included here, but could be found in~\cite{evoc}.

\section{Getting 5-VOCs}\label{sec:fivevertex}
In this section we provide the formulas for computing 5-VOCs derived from \Lem{cut-orb} and also analysis of run time for computing 5-VOCs using this equations. This will also prove \Thm{running-time}. In the run time analysis for computing orbit $\theta_i$, we assume that we have already obtained the counts for $\theta_0$-$\theta_{i-1}$ and all edge orbit counts $E_0$-$E_{11}$. Most of the equations in \Thm{5vertexorbit} have a vertex, an edge, or a triangle as the cut set and are straightforward to follow. we give a proof sketch for the rest of the equations and how they are obtained from \Lem{cut-orb}.

\begin{theorem} \label{thm:5vertexorbit}
For $i \in {0,\ldots,72}$, let $\lambda_i = r(\theta_i)$. Let $TT(u,v)$ denote the count of tailed triangles incident to edge $(u,v)$, where $u$ is the tail vertex ($\theta_9$) and $v$ is in $\theta_{10}$. The value $\cc(u,v,w)$ denotes the number of diamonds ($H_{7}$) that involves the vertices $u,v,w$ such that $u$ and $w$ are the vertices incident to the chord. Let $\cc(u,v)$ be the number of diamonds where $u$ and $v$ are not incident to the chord. And finally, let $W(u,v)$ be the number of wedges between vertices $u$ and $v$. Then, for each vertex $u \in V$,

\begin{flalign*}
\uoMatchSize(u,\lambda_{15}) &=
 \sum_{v \in N(u)} \left[\uoMatchSize(v, \lambda_{4})\right] \\ &- \uoMatchSize(u, \lambda_{5}) -2 \uoMatchSize(u, \lambda_{11}) - 2\uoMatchSize(u, \lambda_{8}) &
\end{flalign*}

\begin{flalign*}
\uoMatchSize(u, \lambda_{16}) &= \uoMatchSize(u, \lambda_{4})(d(u)-1) \\&-\uoMatchSize(u, \lambda_{10}) - 2 \uoMatchSize(u, \lambda_{8}) & 
\end{flalign*}

\begin{flalign*}
\uoMatchSize(u, \lambda_{17}) &= {\uoMatchSize(u, \lambda_{1}) \choose 2} \\&- \uoMatchSize(u, \lambda_{3}) - \uoMatchSize(u, \lambda_{6}) - \uoMatchSize(u, \lambda_{8}) - \uoMatchSize(u, \lambda_{10}) &
\end{flalign*}

\begin{flalign*}
\uoMatchSize(u, \lambda_{18}) &= \sum_{v \in N(u)} [\uoMatchSize(v, \lambda_{6})]  - 3\uoMatchSize(u, \lambda_{7}) - \uoMatchSize(u, \lambda_{10}) \\ & =\sum_v \left[ W(u,v) {d(v)-1 \choose 2} \right] - \uoMatchSize(u, \lambda_{10})  &   
\end{flalign*}

\begin{flalign*}
\uoMatchSize(u, \lambda_{19}) &= \sum_{v \in N(u)} \left[ \uoMatchSize(v, \lambda_{5}) \right] \\&-\uoMatchSize(u, \lambda_{4})-\uoMatchSize(u, \lambda_{u}) - \uoMatchSize(u, \lambda_{10}) &
\end{flalign*}

\begin{flalign*}
\uoMatchSize(u, \lambda_{20}) &= \sum_{v \in N(u)} \left[ (d(u)-1) {d(v)-1 \choose 2} \right] - 2\uoMatchSize(u, \lambda_{10}) &
\end{flalign*}

\begin{flalign*}
\uoMatchSize(u, \lambda_{21}) &=  \sum_{v \in N(u)} \left[ (d(v)-1) {d(u)-1 \choose 2} \right] - \uoMatchSize(u, \lambda_{11}) &
\end{flalign*}

\begin{flalign*}
\uoMatchSize(u, \lambda_{22}) &=   \sum_{v \in N(u)} {d(v)-1 \choose 3} & 
\end{flalign*}

\begin{flalign*}
\uoMatchSize(u, \lambda_{23}) &=  {d(u) \choose 4} & 
\end{flalign*}

\begin{flalign*}
\uoMatchSize(u, \lambda_{24}) &=  \sum_{v \in N(u)} \left[ \uoMatchSize(v, \lambda_{10}) \right] \\ & - \uoMatchSize(u, \lambda_{10}) -2\uoMatchSize(u, \lambda_{11})-2\uoMatchSize(u, \lambda_{12}) &
\end{flalign*}

\begin{flalign*}
\uoMatchSize(u, \lambda_{25}) &=  \sum_{t=(u,v,x)} \left[ (d(v)-2)(d(x)-2) \right] - \uoMatchSize(u, \lambda_{12}) &
\end{flalign*}

\begin{flalign*}
\uoMatchSize(u, \lambda_{26}) &= \sum_{t=(u,v,x)} \left[ (d(u)-2)( (d(x)-2) + (d(v)-2) ) \right] \\& - 2\uoMatchSize(u, \lambda_{13}) &
\end{flalign*}

\begin{flalign*}
\uoMatchSize(u, \lambda_{27}) &= \sum_{v \in N(u)} \left[\uoMatchSize(v, \lambda_{9})\right] -\uoMatchSize(u, \lambda_{11}) -2\uoMatchSize(u, \lambda_{13})\\ & = \sum_v \left[ W(u,v)\uoMatchSize(v, \lambda_{3}) \right] \\&- \uoMatchSize(u, \lambda_{13}) - \uoMatchSize(u, \lambda_{9}) -\uoMatchSize(u, \lambda_{3}) &
\end{flalign*}

\begin{flalign*}
\uoMatchSize(u, \lambda_{28}) &= \sum_{v \in N(u)} \left[ (d(u)-1) \uoMatchSize(v, \lambda_{3}) \right]  \\& - 2\uoMatchSize(u, \lambda_{3}) - 2\uoMatchSize(u, \lambda_{11}) - 2\uoMatchSize(u, \lambda_{12}) &
\end{flalign*}

\begin{flalign*}
\uoMatchSize(u, \lambda_{29}) &= \sum_{t=(u,v,x)} \left[ \uoMatchSize(v, \lambda_{1}) + \uoMatchSize(x, \lambda_{1}) \right] \\& -4\uoMatchSize(u, \lambda_{3}) - \uoMatchSize(u, \lambda_{10}) - 2\uoMatchSize(u, \lambda_{11}) \\& - 2 \uoMatchSize(u, \lambda_{12}) - 2\uoMatchSize(u, \lambda_{13}) &
\end{flalign*}

\begin{flalign*}
\uoMatchSize(u, \lambda_{30}) &= \sum_{v \in N(u)} \left[  (d(v) - 1 ) \uoMatchSize(u, \lambda_{3}) \right] \\& - 2\uoMatchSize(u, \lambda_{3}) - \uoMatchSize(u, \lambda_{10}) - 2\uoMatchSize(u, \lambda_{13}) &
\end{flalign*}

\begin{flalign*}
\uoMatchSize(u, \lambda_{31}) & = \sum_{v \in N(u)} \left[(d(v) -1)\uoMatchSize(v, \lambda_{3})\right] \\& - 2\uoMatchSize(u, \lambda_{9}) - \uoMatchSize(u, \lambda_{10}) - 2\uoMatchSize(u, \lambda_{3}) &
\end{flalign*}

\begin{flalign*}
\uoMatchSize(u, \lambda_{32}) & = \sum_{t=(u,v,x)} \left[ {d(v)-2 \choose 2} + {d(x)-2 \choose 2}\right] &
\end{flalign*}

\begin{flalign*}
\uoMatchSize(u, \lambda_{33}) & = \uoMatchSize(u, \lambda_{3}) {d(u)-2 \choose 2} &
\end{flalign*}

\begin{flalign*}
\uoMatchSize(u, \lambda_{35}) & = \sum_{v \in N(u)}\left[\uoMatchSize(v, \lambda_{8})\right] - 2R_8(u) - R_{13}(u) &
\end{flalign*}

\begin{flalign*}
\uoMatchSize(u, \lambda_{36}) & = \sum_{v} {W(u,v) \choose 2} (d(v) - 2) - \uoMatchSize(u, \lambda_{13}) & 
\end{flalign*}

\begin{flalign*}
\uoMatchSize(u, \lambda_{37}) & = \sum_{v \in N(u)} \left[ E_5(u,v)(d(v) - 2) \right] - 2\uoMatchSize(u, \lambda_{12}) &
\end{flalign*}

\begin{flalign*}
\uoMatchSize(u, \lambda_{38}) & = \uoMatchSize(u, \lambda_{8})(d(u) - 2) - \uoMatchSize(u, \lambda_{13}) &
\end{flalign*}

\begin{flalign*}
\uoMatchSize(u, \lambda_{39}) & = \sum_{v \in N(u)} \left[ \uoMatchSize(v, \lambda_{13}) \right] - \uoMatchSize(u, \lambda_{13}) - 2\uoMatchSize(u, \lambda_{12}) &
\end{flalign*}

\begin{flalign*}
\uoMatchSize(u, \lambda_{40}) & = \sum_{v \in N(u)} E_9(u,v)(d(v) - 3) &
\end{flalign*}

\begin{flalign*}
\uoMatchSize(u, \lambda_{41}) & =  \sum_{v \in N(u)} {E_1(u,v) \choose 2} (d(v)-3) &
\end{flalign*}

\begin{flalign*}
\uoMatchSize(u, \lambda_{42}) & = \sum_{v \in N(u)} {E_1(u,v) \choose 2} (d(u) -3) &
\end{flalign*}

\begin{flalign*}
\uoMatchSize(u, \lambda_{43}) & = \sum_{t=(u,v,x)} \left[(\uoMatchSize(v, \lambda_{3}) -1) + (\uoMatchSize(x, \lambda_{3}) - 1) \right] \\& - 2\uoMatchSize(u, \lambda_{12})  - 2\uoMatchSize(u, \lambda_{13}) &
\end{flalign*}

\begin{flalign*}
\uoMatchSize(u, \lambda_{44}) & = {\uoMatchSize(u, \lambda_{3}) \choose 2} -\uoMatchSize(u, \lambda_{13}) &
\end{flalign*}

\begin{flalign*}
\uoMatchSize(u, \lambda_{45}) & =  \sum_{v \in N(u)} \left[\uoMatchSize(v, \lambda_{12})\right] - 2\uoMatchSize(u, \lambda_{13})-3\uoMatchSize(u, \lambda_{14}) &
\end{flalign*}

\begin{flalign*}
\uoMatchSize(u, \lambda_{46}) & =  \sum_{t=(u,v,x)} \left[ E_7(v,x) \right] - \uoMatchSize(u, \lambda_{11}) - 3\uoMatchSize(u, \lambda_{14}) &
\end{flalign*}

\begin{flalign*}
\uoMatchSize(u, \lambda_{47}) & = \uoMatchSize(u, \lambda_{12})(d(u) -2) - 3\uoMatchSize(u, \lambda_{14}) &
\end{flalign*}

\begin{flalign*}
\uoMatchSize(u, \lambda_{48}) & =  \sum_{t=(u,v,x)} [(E_1(u,v)-1)(d(x)-2) \\& + (E_1(u,x)-1)(d(v)-2)] -6\uoMatchSize(u, \lambda_{14}) &
\end{flalign*}

\begin{flalign*}
\uoMatchSize(u, \lambda_{49}) & = \sum_{v,x \in N(u)} {W(v,x)-1 \choose 2} &
\end{flalign*}

\begin{flalign*}
\uoMatchSize(u, \lambda_{50}) & = \sum_{v} {W(u,v) \choose 3} &
\end{flalign*}

\begin{flalign*}
\uoMatchSize(u, \lambda_{51}) & = \sum_{v} TT(u,v) (W(u,v) - 1) - 2\uoMatchSize(u, \lambda_{12}) &
\end{flalign*}

\begin{flalign*}
\uoMatchSize(u, \lambda_{52}) & = \sum_{t=(u,v,x)} \left[E_5(v,x)\right] - 2\uoMatchSize(u, \lambda_{13}) &
\end{flalign*}

\begin{flalign*}
\uoMatchSize(u, \lambda_{53}) & = \sum_{t=(u,v,x)} \left[E_5(u,v) + E_5(u,x)\right] \\& -2\uoMatchSize(u, \lambda_{13}) - 2\uoMatchSize(u, \lambda_{12}) &
\end{flalign*}

\begin{flalign*}
\uoMatchSize(u, \lambda_{54}) & = \sum_{t=(u,v,x)} {E_1(v,x)-1 \choose 2} &
\end{flalign*}

\begin{flalign*}
\uoMatchSize(u, \lambda_{55}) & = \sum_{v \in N(u)} {E_1(u,v) \choose 3} &
\end{flalign*}

\begin{flalign*}
\uoMatchSize(u, \lambda_{56}) & =  \sum_{v \in N(u)} \left[\uoMatchSize(v, \lambda_{14})\right] - 3\uoMatchSize(u, \lambda_{14}) &
\end{flalign*}

\begin{flalign*}
\uoMatchSize(u, \lambda_{57}) & =  \sum_{v \in N(u)} E_{11}(u,v)(d(v)-3) &
\end{flalign*}

\begin{flalign*}
\uoMatchSize(u, \lambda_{58}) & = \uoMatchSize(u, \lambda_{14}) (d(u) - 3) &
\end{flalign*}

\begin{flalign*}
\uoMatchSize(u, \lambda_{59}) & = \sum_{t=(u,v,x)} \left[ E_9(\langle x,v \rangle) + E_9(\langle v,x \rangle) \right] \\& - 2 \uoMatchSize(u, \lambda_{13}) - 6 \uoMatchSize(u, \lambda_{14}) &
\end{flalign*}

\begin{flalign*}
\uoMatchSize(u, \lambda_{60}) & =  \sum_{v \in N(u)} E_9(\langle v, u \rangle) (E_1(u,v) - 1) -6\uoMatchSize(u, \lambda_{14}) &
\end{flalign*}

\begin{flalign*}
\uoMatchSize(u, \lambda_{61}) & = \sum_{t=(u,v,x)} \left[(E_1(u,v)-1) (E_1(u,x)-1)\right] \\& - 3\uoMatchSize(u, \lambda_{14}) &
\end{flalign*}

\begin{flalign*}
\uoMatchSize(u, \lambda_{62}) & = \sum_{v,x \in N(u)} \left[\cc(v,x)\right] -\uoMatchSize(u, \lambda_{13}) &
\end{flalign*}

\begin{flalign*}
\uoMatchSize(u, \lambda_{63}) & = \sum_{v} \cc(u,v) (W(u,v)-2) &
\end{flalign*}

\begin{flalign*}
\uoMatchSize(u, \lambda_{64}) & = \sum_{v, x \in N(u)} \cc(v,u,x) (W(v,x)-2) &
\end{flalign*}

\begin{flalign*}
\uoMatchSize(u, \lambda_{65}) & =  \sum_{t=(u,v,x)} \left[ E_{11}(v,x)\right] -3\uoMatchSize(u, \lambda_{14}) &
\end{flalign*}

\begin{flalign*}
\uoMatchSize(u, \lambda_{66}) & = \sum_{\left\langle u,v,x,y \right\rangle \text{ is 4-clique}} \left[ E_1(v,x) + E_1(v,y) + E_1(x,y) \right] &
\end{flalign*}

\begin{flalign*}
\uoMatchSize(u, \lambda_{67}) & = \sum_{v \in N(u)} E_{11}(u,v)(E_1(u,v)-2) &
\end{flalign*}

\begin{flalign*}
\uoMatchSize(u, \lambda_{68}) & = \sum_{\substack{v,x \text{ where} \\ \langle u,v,x \rangle \text{ is a wedge}}} {\cc(u,v,x) \choose 2} &
\end{flalign*}

\begin{flalign*}
\uoMatchSize(u, \lambda_{69}) & = \sum_{v,x \in N(u)} {\cc(u,v,x) \choose 2}&
\end{flalign*}

\begin{flalign*}
\uoMatchSize(u, \lambda_{70}) & =  \sum_{\left\langle u,v,x,y \right\rangle \text{ is 4-clique}} K_4(v, x, y) - 1 &
\end{flalign*}

\begin{flalign*}
\uoMatchSize(u, \lambda_{71}) & = \sum_{t=(u,v,x)} {K_4(t) \choose 2} &
\end{flalign*}

\end{theorem}

\begin{proof}\label{5vertexorbit-proof}
For orbits $\theta_{36}$, $\theta_{49}$, $\theta_{50}$, $\theta_{51}$, $\theta_{63}$, and $\theta_{64}$, we need the counts of wedges which have the vertex at hand in the middle as in the equation for $\theta_{49}$, or at one of the ends, as in the equation for $\theta_{36}$. We do not precompute and store these counts for all vertices as it could be expensive. But we can get these counts while counting, by enumerating wedges in time $O(W(G))$~\cite{PiSeVi17}. In the equation for $\theta_{51}$, we need the counts of $TT(u,v)$. But this is easy to get while enumerating the wedges between $u$ and $v$, and using the triangle per-edge counts for edge $(x,v)$, where $(u,x,v)$ is a wedge. Equation of orbits $\theta_{62}$, $\theta_{63}$, $\theta_{64}$, $\theta_{68}$, and $\theta_{69}$ require the counts of diamonds. These counts are too expensive to precompute and store for all the vertices, so we do it while computing the counts for each vertex, using triangle counts for each edge. To compute the coutns of $\theta_{70}$ and $\theta_{71}$, we need to use the counts of 4-cliques incident to each triangle $t$, which we can get in  $O(W(G) + \cc(G) + m + n)$~\cite{PiSeVi17}.
\end{proof}

Finally, we can prove \Thm{running-time}.

\begin{proof}[Proof of \Thm{running-time}]
By \Thm{4orbit}, we know that we can obtain all the counts for orbits $\theta_0$-$\theta_{14}$ and $E_0$-$E_{11}$ in time $O(W(G) + \cc(G) + m + n)$. Also for each triangle $t$, we can get $K_4(t)$ and for each edge $e$, the list of triangles incident to $e$ in time $O(W(G) + \cc(G) + m + n)$~\cite{PiSeVi17}. 
We need to show that computing orbit counts for $\theta_{15}$-$\theta_{72}$ takes time $O(W(G) + \cc(G) + \dpath(G^\rightarrow) + \dbp(G^\rightarrow) + m + n)$. By \Thm{5-Cycle-Clique_running-time}, $\theta_{34}$ (the only orbit in 5-cycle) counts can be obtained in time $O(W(G) + \dpath(G^\rightarrow) + m + n)$ and the counts for $\theta_{72}$ (the only orbit in 5-clique) takes $O(\dpath(G^\rightarrow))$ to compute~\cite{PiSeVi17}. So, we only need to show that computing 5-VOCs except 5-cycle and 5-clique, using each of the equations in \Thm{5vertexorbit} takes time $O(W(G) + \cc(G) + \dpath(G^\rightarrow) + \dbp(G^\rightarrow) + m + n)$.

We divide the set of orbits of 5-vertex patterns to categories with different runtime. When analysing the runtime for equation Of $\theta_i$, we assume that we have access to the counts for $\theta_0$-$\theta_{i-1}$ and all edge orbit counts $E_0$-$E_{11}$, as we have stored them previously. Orbit in each category are shown in \Tab{orbits-category}.

\begin{asparaitem}
\item Orbits that we can count in time $O(n)$ for all vertices:

Computing these vertex orbits, we only need to pass over vertices in $G$, and then it is straightforward to get the counts for each vertex in constant time using the equations in \Thm{5vertexorbit}.

\item Orbits that we can count in time $O(m+n)$ for all vertices:
In this category, to compute the counts for each vertex in $G$, we enumerate its neighborhood. This takes time $O(m+n)$ overall.

\item Orbits that we can count in time $O(W(G) + m + n)$:
Enumerating all the wedges suffices to compute the counts for $\theta_{36}$, $\theta_{49}$, and $\theta_{50}$ using their equations. While enumerating wedges to get the counts of orbit $\theta_{51}$ for vertex $u$, we need the count of tailed triangles incident to edge $(u,v)$, where $u$ is the tail vertex ($\theta_9$) and $v$ is in $\theta_{10}$. But this is easy to get using triangle counts, while $(u,v)$ is the wedge at hand during the wedge enumeration.

The rest of the orbits in this category could be obtained by enumerating all the triangles, which is possible in $O(W(G))$.

\item Orbits that we can count in time $O(W(G) + \cc(G) + m + n)$:
For $\theta_{66}$ and $\theta_{70}$, we need to enumerate 4-cliques, which takes time $O(W(G) + \cc(G) + m + n)$~\cite{PiSeVi17}. Getting counts of orbit $\theta_{71}$ requires enumeration of triangles, but for each triangle $t$ at hand, we need to get $K_4(t)$, which is overall possible in time $O(W(G) + \cc(G) + m + n)$.

To get the rest of the orbit counts in this category, we need to enumerate diamonds. Similar to the way we enumerate wedges while enumerating neighbors of a vertex, instead of precomputing and storing all the wedge counts, we enumerate diamonds while enumerating wedges, using triangle counts that we already have.

\end{asparaitem}

\end{proof}

\begin{table}[t]
\caption{\small Time for computing 5-VOCs for all vertices using equations in \Thm{5vertexorbit}}
\begin{center}
\begin{tabular}{c|c}
5-VOC runtime & Orbits\\\hline
$O(n)$ & $\theta_{16}$, $\theta_{17}$, $\theta_{23}$, $\theta_{33}$, $\theta_{38}$, $\theta_{44}$, $\theta_{47}$, $\theta_{58}$ \\\hline
\multirow{3}{*}{$O(m+n)$} & $\theta_{15}$, $\theta_{18}$, $\theta_{19}$, $\theta_{20}$, $\theta_{21}$, $\theta_{22}$, $\theta_{24}$, $\theta_{27}$\\ & $\theta_{28}$, $\theta_{30}$, $\theta_{31}$, $\theta_{35}$, $\theta_{37}$, $\theta_{39}$, $\theta_{40}$, $\theta_{41}$,\\& $\theta_{42}$, $\theta_{45}$, $\theta_{55}$, $\theta_{56}$, $\theta_{57}$, $\theta_{60}$, $\theta_{67}$ \\\hline
\multirow{3}{*}{$O(W(G) + m + n)$} & $\theta_{25}$, $\theta_{26}$, $\theta_{29}$, $\theta_{32}$, $\theta_{36}$, $\theta_{43}$, $\theta_{46}$, $\theta_{48}$,\\& $\theta_{49}$, $\theta_{50}$, $\theta_{51}$, $\theta_{52}$, $\theta_{53}$, $\theta_{54}$, $\theta_{59}$, $\theta_{61}$, $\theta_{65}$ \\
\hline
\end{tabular}
\end{center}
\label{tab:orbits-category}
\end{table}

\begin{figure}[b]
\[ A=
\left(
\begin{array}{ccccccccccc}

1 & 0 & 0 & 0 & 2 & 2 & 1 & 0 & 4 & 2 & 6 \\
0 & 1 & 0 & 0 & 2 & 0 & 1 & 2 & 2 & 2 & 6 \\
0 & 0 & 1 & 0 & 0 & 1 & 1 & 0 & 2 & 1 & 3 \\
0 & 0 & 0 & 1 & 0 & 0 & 0 & 1 & 0 & 1 & 1 \\
0 & 0 & 0 & 0 & 1 & 0 & 0 & 0 & 1 & 1 & 3 \\
0 & 0 & 0 & 0 & 0 & 1 & 0 & 0 & 2 & 0 & 3 \\
0 & 0 & 0 & 0 & 0 & 0 & 1 & 0 & 2 & 2 & 6 \\
0 & 0 & 0 & 0 & 0 & 0 & 0 & 1 & 0 & 2 & 3 \\
0 & 0 & 0 & 0 & 0 & 0 & 0 & 0 & 1 & 0 & 3 \\
0 & 0 & 0 & 0 & 0 & 0 & 0 & 0 & 0 & 1 & 3 \\
0 & 0 & 0 & 0 & 0 & 0 & 0 & 0 & 0 & 0 & 1 \\

\end{array}
\right)
\]
\caption{Matrix transforming induced vertex orbit counts for orbits 0-14 to non-induced counts}
\label{fig:A}
\end{figure}

\begin{figure}[b]
\[
A^{-1}=
\left(
\begin{array}{ccccccccccc}
1 & 0 & 0 & 0 & -2 & -2 & -1 & 0 & 4 & 2 & -6 \\                                                                                                                                      
0 & 1 & 0 & 0 & -2 &  0 & -1 & -2 & 2 & 6 & -12 \\                         
0 & 0 & 1 & 0 & 0 & -1 & -1 & 0 & 2 & 1 & -3 \\
0 & 0 & 0 & 1 & 0 & 0 & 0 & -1 & 0 & 1 & -1 \\
0 & 0 & 0 & 0 & 1 & 0 & 0 & 0 & -1 & -1 & 3 \\
0 & 0 & 0 & 0 & 0 & 1 & 0 & 0 & -2 & 0 & 3 \\
0 & 0 & 0 & 0 & 0 & 0 & 1 & 0 & -2 & -2 & 6 \\
0 & 0 & 0 & 0 & 0 & 0 & 0 & 1 &  0 & -2 & 3 \\
0 & 0 & 0 & 0 & 0 & 0 & 0 & 0 & 1 & 0 & -3 \\
0 & 0 & 0 & 0 & 0 & 0 & 0 & 0 & 0 & 1 & -3 \\
0 & 0 & 0 & 0 & 0 & 0 & 0 & 0 & 0 & 0 & 1 \\
\end{array}
\right)
\]
\caption{Matrix transforming non-induced vertex orbit counts for orbits 0-14 to induced counts}
\label{fig:Ainv}
\end{figure}

\end{document}